\documentclass[10pt,twocolumn,twoside]{IEEEtran}

\usepackage[cmex10]{amsmath}
\usepackage{amssymb}
\usepackage{graphicx}
\usepackage{cite}
\usepackage{cases}
\usepackage{amssymb,amsmath}
\usepackage{color}
\usepackage{glossaries}
\usepackage{amsthm}
\usepackage{lettrine}

\usepackage[textwidth=2cm, textsize=small]{todonotes}

\ifCLASSOPTIONcompsoc
	\usepackage[caption=false,font=normalsize,labelfont=sf,textfont=sf]{subfig} 
\else
	\usepackage[caption=false,font=footnotesize]{subfig}
\fi

\renewcommand{\vec}[1]{\mathbf{#1}}
\providecommand{\norm}[1]{\lVert#1\rVert}

\newtheorem{theorem}{Theorem}
\newtheorem{corollary}[theorem]{Corollary}
\newtheorem{lemma}[theorem]{Lemma}

\newacronym{ild}{ILD}{interaural level difference}%
\newacronym{itd}{ITD}{interaural time difference}%
\newacronym{hrtf}{HRTF}{head related transfer function}%
\newacronym{wfs}{WFS}{wave-field synthesis}%
\newacronym{hoa}{HOA}{higher-order ambisonics}%
\newacronym{oct}{OCT}{optimized cardioid triangle}%
\newacronym{psr}{PSR}{perceptual sound-field reconstruction}%
\newacronym{vbap}{VBAP}{vector-base amplitude panning}%
\newacronym{idir}{ID}{intensity directivity}%
\newacronym{ictd}{ICTD}{inter-channel time difference}%
\newacronym{icld}{ICLD}{inter-channel level difference}%
\newacronym{tild}{TILD}{time-intensity linear directivity}%
\newacronym{tid}{TID}{time-intensity directivity}%
\newacronym{erb}{ERB}{equivalent rectangular bandwidth}%
\newacronym{wng}{WNG}{white noise gain}%
\newacronym{snr}{SNR}{signal to noise ratio}%
\newacronym{rir}{RIR}{room impulse response}%
\newacronym{sdn}{SDN}{scattering delay network}%
\newacronym{dwm}{DWM}{digital waveguide mesh}%
\newacronym{dwn}{DWN}{digital waveguide network}%
\newacronym{fdn}{FDN}{feedback delay network}%
\newacronym{ism}{IM}{image method}
\newacronym{ortf}{ORTF}{office de radiodiffusion t\'el\'evision fran\c{c}aise}
\newacronym{fft}{FFT}{fast Fourier transform}
\newacronym{flops}{FLOPS}{floating point operations per second}
\newacronym{ned}{NED}{normalized echo density}
\newacronym{fdtd}{FDTD}{finite-difference time-domain}
\newacronym{bibo}{BIBO}{bounded-input bounded-output}
\newacronym{svd}{SVD}{singular value decomposition}
\newacronym{vr}{VR}{virtual reality}
\newacronym{cave}{CAVE}{CAVE automatic virtual environment}

\begin{document}
\title{Efficient Synthesis of Room Acoustics  \\via Scattering Delay Networks}%
\author{Enzo De Sena,~\textit{Member,~IEEE}, 
		H\"{u}seyin~Hac\i{habib}o\u{g}lu,~\textit{Senior Member,~IEEE}, \\
		Zoran Cvetkovi\'c,~\textit{Senior Member,~IEEE}, 
		Julius O. Smith,~\textit{Member,~IEEE}%

\thanks{Copyright \copyright~2015 IEEE. Personal use of this material is permitted. However, permission to use this material for any other purposes must be obtained from the IEEE by sending a request to pubs-permissions@ieee.org. http://ieeexplore.ieee.org/xpl/articleDetails.jsp?arnumber=7113826}
\thanks{Manuscript received February 19, 2015; accepted May 16, 2015. Date of publication June 01, 2015. 
The work reported in this paper was partially funded by (i) EPSRC Grant EP/F001142/1, (ii) European Commission under Grant Agreement no. 316969 within the FP7-PEOPLE Marie Curie Initial Training Network ``Dereverberation and Reverberation of Audio, Music, and Speech (DREAMS)", and (iii) TUBITAK Grant 113E513 ``Spatial Audio Reproduction Using Analysis-based Synthesis Methods".
The method presented in this paper is protected by USPTO Patent n. 8908875.The associate editor coordinating the review of this manuscript and approving it for publication was Prof. Bozena Kostek.}
\thanks{Enzo De Sena is with ESAT--STADIUS, KU Leuven,
Kasteelpark Arenberg 10, 3001 Leuven, Belgium (e-mail: enzo.desena@esat.kuleuven.be). This work was done in part while he was with CTR, King's College London.}%
\thanks{Zoran Cvetkovi\'c is with the Centre for Telecommunications Research (CTR), King's College London, Strand, London, WC2R 2LS, United Kingdom (e-mail: zoran.cvetkovic@kcl.ac.uk).}%
\thanks{H\"useyin Hac\i{}habibo\u{g}lu is with Department of Modeling and Simulation, Informatics Institute, Middle East Technical University, Ankara, 06800, Turkey (e-mail: hhuseyin@metu.edu.tr). This work was done in part while he was with King's College London.}%
\thanks{Julius O. Smith III is with CCRMA, 
Stanford University, Stanford, CA 94304, USA.}%
\thanks{This paper has supplementary downloadable material available at http://ieeexplore.ieee.org, provided by the author. The material includes several audio samples generated using the proposed room acoustic simulator. Contact enzo.desena@esat.kulueven.be for further questions about this work.}%
}

\graphicspath{{figures/}}

\IEEEpubid{Copyright \copyright~2015 IEEE}%
\markboth{IEEE/ACM TRANSACTIONS ON AUDIO, SPEECH, AND LANGUAGE PROCESSING, VOL. 23, NO. 9, SEPTEMBER 2015}{De Sena \MakeLowercase{\textit{et al.}}: Efficient Synthesis of Room Acoustics via Scattering Delay Networks}%
\maketitle 

\setlength{\arraycolsep}{0.0em}

\begin{abstract}
An acoustic reverberator consisting of a network of delay lines connected via  scattering junctions is
proposed.  All parameters of the reverberator are derived 
from physical properties of the enclosure it simulates. It allows for
simulation of unequal and frequency-dependent wall absorption,
as well as directional sources and microphones.
The reverberator renders the first-order reflections exactly, while
making progressively coarser approximations of higher-order
reflections.  The rate of energy decay is close to that obtained
with the \gls{ism} and consistent with the predictions of
Sabine and Eyring equations. The time evolution of the normalized
echo density, which was previously shown to be correlated with the
perceived texture of reverberation, is also close to that of \gls{ism}. 
However, its computational complexity
is one to two orders of magnitude lower, 
comparable to the computational complexity of a \gls{fdn}, 
and its memory requirements are negligible.

\end{abstract}

\glsreset{fdn}
\glsreset{ism}

\section{Introduction}

\lettrine[lines=2]{A}{} comprehensive  account of the first fifty years 
of artificial reverberation 
 in~\cite{valimaki2012fifty}
identifies three main classes of digital reverberators:
delay networks, physical room models and convolution-based algorithms. 
The earliest class consisted of delay networks, 
which were the only artificial reverberators feasible 
with the integrated circuits of the time.
The first delay network reverberator, as introduced by Schroeder, was  a cascade of several allpass filters, a
parallel bank of feedback comb filters and a mixing matrix
\cite{schroder1961,schroder1962,valimaki2012fifty}. 
Since then, a large number of delay networks have been proposed and used commercially.
Most of these networks were designed heuristically and by trial and error.
\Glspl{fdn} were developed on a more solid scientific grounding as a multichannel 
extension of the Schroeder reverberator~\cite{gerzon1976unitary,stautner1982},
and  consist of 
parallel delay lines connected recursively through a 
unitary feedback matrix. The state-of-the-art \gls{fdn} is due to 
Jot and Chaigne~\cite{jot1991}, who proposed
using delay lines connected in series with multiband absorptive filters to obtain a frequency-dependent
reverberation time. \Gls{fdn} reverberators
are still among the most commonly used 
artificial reverberators owing to their simple design, 
extremely low computational complexity and high reverberation quality.
While no new standard seems to have emerged yet, 
a number of more intricate networks have been proposed more recently 
that show improvements over \glspl{fdn}, 
sometimes even in terms of computational complexity~\cite{huopaniemi1997efficient,dahl2000reverberator,lee2012switched}.

\Glspl{fdn} are structurally equivalent to a particular case of \glspl{dwn}~\cite{Rocchesso:2002rq}, 
reverberators based on the concept of digital waveguides, 
 introduced in~\cite{Smith-III:1985ib}. 
\glspl{fdn} and \glspl{dwn} can also be viewed both as 
networks of multiport elements, as explained by Koontz in~\cite{koontz2013multiport}.
\Glspl{dwn} 
consist of a closed network of bidirectional delay lines
interconnected at lossless junctions.
\Glspl{dwn} have an exact physical interpretation 
as a network of interconnected acoustical tubes and
have a number of appealing properties 
in terms of computational efficiency and numerical stability.
Reverberators based on delay networks have been 
widely used for artistic purposes in music production. 
High-level interfaces enable artists and sound engineers to adjust the available free parameters until the intended qualities of reverberated sound are achieved. 

In the domain of predictive architectural modeling, 
virtual reality and computer games, on the other hand, 
the objective of artificial reverberators is to emulate the 
response of a physical room given 
a set of physically relevant parameters~\cite{vorlander2008auralization,SaviojaEtAl95}.
These parameters include, for instance, the room geometry, 
absorption and diffusion characteristics of walls and objects, 
and position and directivity pattern 
of source and microphone.  

\IEEEpubidadjcol

Various physical models have been proposed 
in the past for the purpose 
of room acoustic synthesis. Widely used 
geometric-acoustic  models  make the 
simplifying assumption that sound waves propagate as rays.
In particular, the ray tracing approach explicitly tracks 
the rays emitted by the acoustic source as 
they bounce off the  surfaces of  a modeled enclosure. 
The \gls{ism} is an alternative algorithmic implementation 
that formally replaces the physical boundaries surrounding 
the source with an equivalent infinite lattice of image sources~\cite{Allen79b,borish1984}. 
Allen and Berkley proved that, in the case of rectangular rooms,
this approach is equivalent 
to solving the wave equation provided that  the walls are perfectly rigid~\cite{Allen79b}. 
When the walls are not rigid, the results of the \gls{ism} are no longer 
physically accurate, but the method still retains its 
geometric-acoustic interpretation. 
The \gls{ism} can also be used to model the acoustics in arbitrary polyhedra, 
as described by Borish in~\cite{borish1984}.
The main advantage of geometric-acoustic models 
in comparison to other physical models 
is their lower -- although still considerable -- computational complexity.
However, they do not model important 
wave phenomena such as diffraction and interference.
These phenomena are inherently modeled by 
methods based on time and space discretization of the wave equation, 
such as \gls{fdtd} and \gls{dwm} models~\cite{Bilbao:2003ue,murphy2007acoustic,karjalainen2004digital,SaviojaEtAl95}.
The main limitation of physical room models consists of their 
significant computational and memory requirements.
While this may not be problematic for predictive acoustic modeling applications,
it is a significant limitation for interactive applications, such as \gls{vr}.
Similarly, convolutional methods, 
which operate by filtering anechoic audio samples 
with measured \glspl{rir}, 
do not allow interactive operation unless interpolation among an extensive number of \glspl{rir} is supported.

Virtual reality has become a widespread technology, 
with applications in military training, immersive virtual musical instruments, 
archaeological acoustics, 
and, in particular, computer games. 
Along with realistic 
graphics rendering, spatial audio is one of the most important factors
that affect how convincing its users perceive a virtual 
environment~\cite{zyda2005visual}. 
The aural and visual components should 
be consistent and mutually supportive so 
as to minimise cross-modal sensory conflicts~\cite{zyda2005visual}.

Room acoustic synthesis for \gls{vr} 
also requires flexibility in terms of audio output devices.
For example, a full scale \gls{vr} suite such 
as a \gls{cave} can use
Ambisonics~\cite{Gerzon:1985rw} 
or wave field synthesis (WFS), which requires from tens to
hundreds of loudspeaker channels~\cite{Boone:1995qa}. 
In contrast, a 
portable game console has a two-channel output that may be
used, for instance, to reproduce binaural audio over headphones~\cite{Moller:1991kb}. 
Typical home users, on the other hand, 
commonly use stereophonic, 5.1 or 7.1 setups, whereas the ultra high definition TV (UHDTV) standard, also aimed at home users, makes provisions for 22.2 setups~\cite{SMPTE2008}.

In summary, room acoustic synthesizers for \gls{vr} require 
(i)~explicit (tuning-free) modeling of a given virtual space,
(ii)~scalability in terms of playback configuration, 
and (iii)~low computational complexity.
On the other hand, of the three classes of digital reverberators, 
(i)~convolutional methods do not allow interactive operation without extensive tabulation and interpolation,
(ii)~delay network methods do not model explicitly a given 
virtual space, 
and (iii)~physical models have a high computational cost.

In order to combine the appealing properties of delay networks
and physical models, one possible approach consists of 
designing delay networks that have parameters with 
an explicit physical interpretation.
Studies in 
\cite{rocchesso1995ball,Rocchesso:2002rq} and \cite{huangdigital}
follow this direction.
In  \cite{Rocchesso:2002rq} and \cite{rocchesso1995ball} 
the length of the delay lines of  \glspl{fdn} are chosen 
such that the lowest eigenfrequencies of the room are reconstructed exactly. 
In~\cite{huangdigital}, Karjalainen \textit{et al.} use a \gls{dwn} with few junctions 
to model a rectangular room.
The rationale behind the design is 
to aggressively prune-down a \gls{dwm}
in order to reduce the complexity 
while retaining an acceptable perceptual result.
This approach has various advantages but requires 
careful manual tuning in order to provide satisfactory results~\cite{huangdigital}.

In~\cite{de2011scattering} and~\cite{hacihabiboglu2011frequency}, following the same concept of \gls{dwn} 
structures as studied by Karjalainen \textit{et al.} in~\cite{huangdigital}, 
we presented an architecture that has a number of 
appealing properties.
The proposed architecture, which we refer to as \gls{sdn},
renders the direct-path component 
and first-order early reflections accurately both in time and amplitude, 
while producing a progressively coarser approximation of higher-order reflections. 
For this reason, \gls{sdn} can be interpreted 
as an approximation 
to geometric acoustics models 
\cite{Allen79b,borish1984}.
\Glspl{sdn} thus approach the accuracy of
full-scale room simulation while maintaining computational efficiency
on par with typical delay network-based methods. 
Furthermore, the parameters of \gls{sdn} are inherited directly 
from room geometry and absorption properties of wall materials, 
and therefore do not require ad hoc tuning.

This paper further explores and completes 
the design presented in~\cite{de2011scattering} and~\cite{hacihabiboglu2011frequency}. 
All design choices are now explained on a physical basis.
Furthermore, the paper includes 
a theoretical analysis of optimal scattering matrices, 
a comparison with the \gls{ism} in terms of reverberation time 
and normalized echo density~\cite{huang2007aspects}, 
an analysis of the computational complexity and of the memory requirements, 
and an analysis of the modal density~\cite{gardner1998}. 
The paper is organized as follows. Section~\ref{sec:bgrd} presents a
brief overview of \glspl{fdn}, \glspl{dwn} and models 
proposed by Karjalainen \textit{et al.}~\cite{huangdigital}.
Section~\ref{sec:dwmfdn} describes the proposed \gls{sdn} 
method. 
The properties of \glspl{sdn} are studied in Section~\ref{sec:propsdn}.
Section~\ref{sec:results} presents numerical evaluation results.
Section~\ref{sec:conclusions} concludes the paper.

\section{Background}\label{sec:bgrd}

The proposed \gls{sdn} reverberator draws inspiration from \gls{dwn} and \gls{dwm} structures, 
however it is in essence a recursive linear time-invariant system. 
Hence, to provide a comprehensive context, in this section we  briefly review \glspl{fdn},
which are the most commonly used recursive linear time-invariant reverberators, followed by  a more detailed review of relevant \gls{dwn} and \gls{dwm} material.

\subsection{Feedback Delay Networks}
\begin{figure}
\centering
\includegraphics[width=0.5\textwidth]{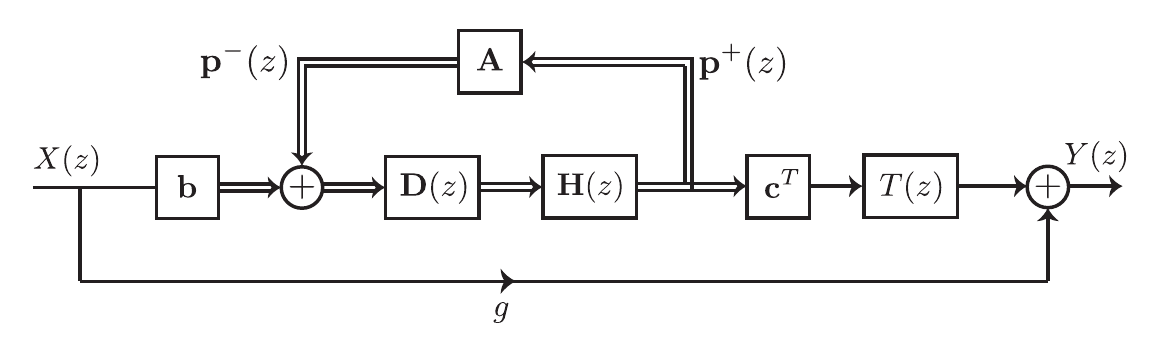}
\caption{Block diagram of the modified \protect\gls{fdn} reverberator as proposed by Jot and Chaigne~\cite{jot1991}.}
\label{fig:jotblock}
\end{figure}
\glsreset{fdn}

The canonical \gls{fdn} form, as proposed by Jot and Chaigne~\cite{jot1991}, is shown in
\figurename~\ref{fig:jotblock}.
Here, 
$\mathbf{b}$ and $\mathbf{c}$ are input and output gains, respectively, 
$\mathbf{D}(z)=\textrm{diag}(z^{-m_{1}},\ z^{-m_{2}},\ \dots,\ z^{-m_{N}})$ 
are integer delays, $\mathbf{H}(z)=\textrm{diag}(H_1(z),\dots,\ H_N(z))$ 
are absorption filters, $T(z)$ is the tone correction filter,  
$g$ is the gain of the direct path, 
and $\mathbf{A}$ is the  feedback matrix.
The absorption filters can be designed so as to obtain a desired 
reverberation time in different frequency bands~\cite{jot1991}, 
or to match those calculated from a measured \gls{rir}~\cite{jot1992analysis}.
To achieve  a high-quality late reverberation, the  feedback loop should be lossless, i.e. energy-preserving, hence typically the feedback
matrix  $\mathbf{A}$ is designed to be unitary. 
Each particular choice of
the feedback matrix has corresponding implications on subjective or objective qualities 
of the reverberator~\cite{fallerunitary}; e.g. in the particular case of the identity matrix, $\mathbf{A}=\mathbf{I}$, 
the \gls{fdn} structure reduces  to $N$ 
comb filters connected in parallel and acts as the 
Schroeder reverberator~\cite{schroder1962}.
Note, however that 
unitary matrices are only a subset of possible lossless feedback matrices~\cite{Rocchesso:2002rq,jot1997efficient};
we elaborate on this point in the next subsection.

\glsreset{dwm}
\subsection{Digital Waveguide Networks}
\label{sec:dwn}

\begin{figure}[t]
\centering
\includegraphics[width=0.45\textwidth]{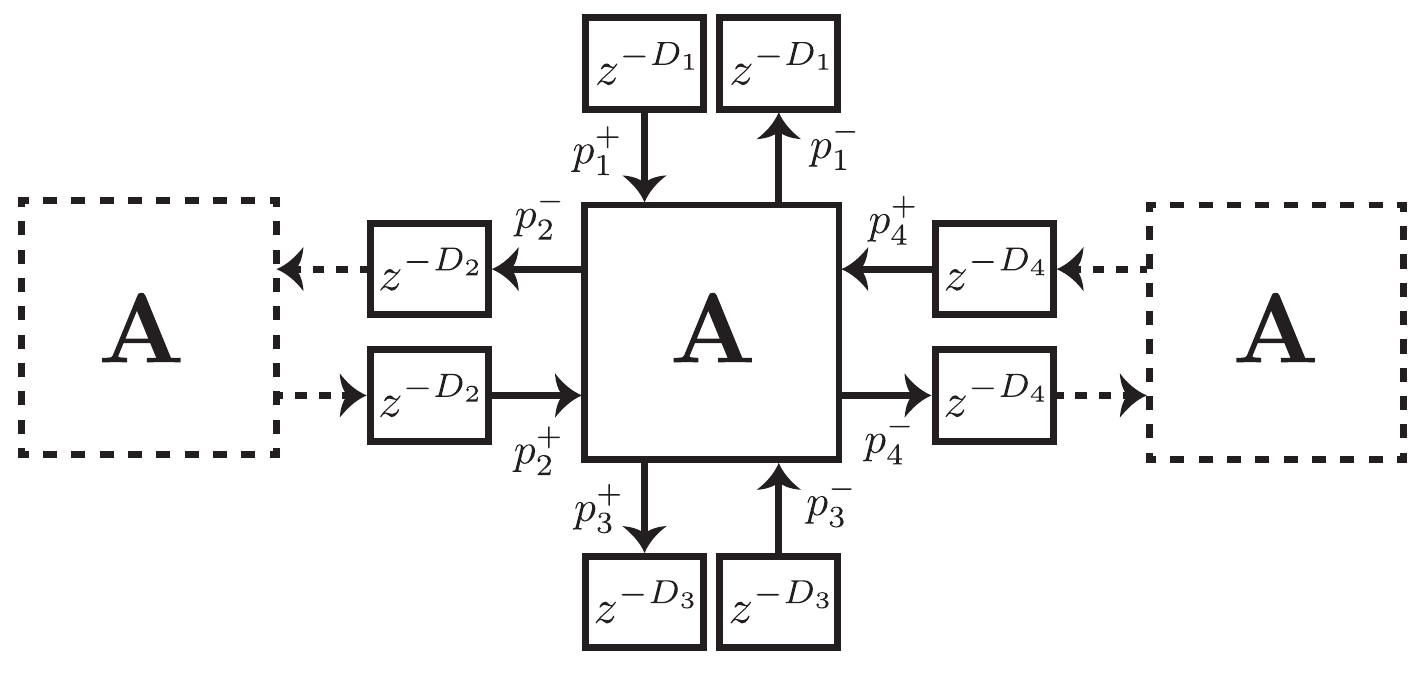}
\caption{Operation of a \protect\gls{dwn} around a junction with $K=4$
 waveguides.}
\label{fig:dwmjunction}
\end{figure}

\Glspl{dwn} consist of a closed network 
of  digital waveguides~\cite{Smith-III:1985ib}.
A digital waveguide is made up of a pair of delay lines,
which implement the digital equivalent 
of the d'Alembert solution of the wave equation
in a one-dimensional medium.
The digital waveguides are interconnected at junctions,
characterised by corresponding scattering matrices.
\figurename\ \ref{fig:dwmjunction} shows an example
of four digital waveguides with length $D_1,...,D_4$ samples
that meet at a junction with scattering matrix $\mathbf{A}$.
In  general, a junction scatters incoming wave variables
$\mathbf{p}^+=\left[p_{1}^+,\dots,p_{K}^+\right]^T$ to produce outgoing wave variables
$\mathbf{p}^-=\left[p_{1}^-,\dots,p_{K}^-\right]^T$ 
according to
$\mathbf{p}^-=\mathbf{A}\mathbf{p}^+$. 
Note that if all digital waveguides are terminated by an ideal non-inverting reflection, the \gls{dwn} is structurally equivalent 
to the feedback loop of  an \gls{fdn} with feedback matrix $\mathbf{A}$ and delay-line lengths of 
$2D_1,...,2D_4$ samples~\cite{Rocchesso:2002rq}.

\Gls{dwn} junctions are lossless. In this context, \emph{losslessness} is defined according to 
classical network theory~\cite{smith2010physical}. 
In particular, 
a junction with scattering matrix $\mathbf{A}$ is said to be lossless if 
the input and output \emph{total complex power} are equal:
\begin{equation}
\mathbf{p}^{+*}\mathbf{Y}\mathbf{p}^+=\mathbf{p}^{-*}\mathbf{Y}\mathbf{p}^- \Rightarrow \mathbf{A}^*\mathbf{Y}\mathbf{A}=\mathbf{Y}
\label{eq:lossless}
\end{equation}
where 
 $\mathbf{Y}$ is a Hermitian 
positive-definite matrix~\cite{Rocchesso:2002rq} and $(\cdot)^*$ denotes the conjugate transpose.
The quantity $\mathbf{p}^{\pm *}\mathbf{Y}\mathbf{p}^\pm$
is the square of the \emph{elliptic norm} of $\mathbf{p}^\pm$
induced by $\mathbf{Y}$. 
It can be shown that
a matrix $\mathbf{A}$ is lossless if and only if 
its eigenvalues lie on the unit circle and it admits a basis of linearly independent vectors~\cite{Rocchesso:2002rq}.
A consequence of this result is that lossless feedback matrices can be 
fully parametrized as
$\mathbf{A}=\mathbf{T}^{-1}\mathbf{\Lambda}\mathbf{T}$, 
where $\mathbf{T}$ is any invertible matrix 
and $\mathbf{\Lambda}$ is any unit-modulus diagonal matrix~\cite{Rocchesso:2002rq}.

The \gls{dwn} can also be interpreted as a physical model for a network of acoustic tubes.
In this case $\mathbf{A}$ assumes a particular form.
If we denote by $y_i$ the characteristic 
admittance of the $i$-th tube
and by $v_i$ the volume velocity of the $i$-th tube at the junction,
the continuity of pressure and conservation of velocity
at the junction give, from  \cite{smith2010physical}:
\begin{align}
p_1&=p_2=\dots=p_K=p, \label{eq:continuity}\\
v_1&+v_2+\dots+v_K = 0 \label{eq:conservation},
\end{align}
respectively, where $p_i=p_i^+ +p_i^-$ denotes 
the acoustic pressure of the $i$-th tube.
Equations (\ref{eq:continuity}) and (\ref{eq:conservation})
imply that the  pressure $p$ at the junction is given by 
\begin{equation}
p=\frac{2}{\sum_{i=1}^{K}y_i}\sum_{i=1}^{K}y_i p_i^+=\frac{2}{\sum_{i=1}^{K}y_i}\sum_{i=1}^{K}y_i p_i^-,
\label{eq:pressurep}
\end{equation}
where we used $v_i=v_i^+ + v_i^-$ 
and  Ohm's law for traveling waves $v_i^+=y_i p_i^+$ 
and $v_i^-=-y_i p_i^-$~\cite{smith2010physical}.
Since $p_i^-=p - p_i^+$, the scattering matrix can be expressed as
\begin{equation}
\mathbf{A} = \frac{2}{\langle\mathbf{1},\mathbf{y}\rangle} \mathbf{1} \mathbf{y}^T-\mathbf{I},
\label{matrix:unnorm}
\end{equation}
where $\mathbf{1}=[1,\dots,1]^T$,  $\mathbf{y}=[y_1,\dots,y_K]^T$, $\langle\cdot,\cdot\rangle$ denotes the scalar product, 
and $\mathbf{I}$ is the identity matrix.
Observe that the scattering matrix  in (\ref{matrix:unnorm}) 
satisfies equation~(\ref{eq:lossless})
with $\mathbf{Y}=\text{diag}\{y_1,\dots,y_K\}$
and is therefore lossless.
In this physically-based case,
the square of the elliptic norm of $\mathbf{p}^\pm$ induced by $\mathbf{Y}$  has  the meaning of incoming/outgoing acoustic power:
$\mathbf{p}^{\pm *}\mathbf{Y}\mathbf{p}^\pm = 
\sum_{i=1}^K y_i|p_i^\pm|^2=
\pm\sum_{i=1}^K v_i^*p_i^\pm$~\cite{Rocchesso:2002rq}.

An equivalent formulation of \glspl{dwn} involves 
normalized pressure waves, defined as $\widetilde{p}_i^\pm=p_i^\pm\sqrt{y_i}$.
In this case, 
the propagating wave variables  $\widetilde{p}_i^\pm$ represent
the square root of the traveling signal power~\cite{Rocchesso:2002rq}. 
If we define $\widetilde{\mathbf{Y}}=\text{diag}(\sqrt{y_1},\dots,\sqrt{y_K})$, 
the normalized output wave can be written as 
$\widetilde{\mathbf{p}}^-=\widetilde{\mathbf{Y}} \mathbf{p}^-=\widetilde{\mathbf{Y}}\mathbf{A}\mathbf{p}^+=\widetilde{\mathbf{Y}} \mathbf{A}\widetilde{\mathbf{Y}}^{-1}\widetilde{\mathbf{p}}^+=\widetilde{\mathbf{A}}\widetilde{\mathbf{p}}^+$,
where  $\widetilde{\mathbf{A}}=\widetilde{\mathbf{Y}} \mathbf{A}\widetilde{\mathbf{Y}}^{-1}$.
The equivalent scattering matrix $\widetilde{\mathbf{A}}$ can  be expressed as
\begin{equation}
\widetilde{\mathbf{A}} = \frac{2}{\|\widetilde{\mathbf{y}}\|^2}\widetilde{\mathbf{y}}\widetilde{\mathbf{y}}^T-\mathbf{I},
\label{eq:house}
\end{equation}
which is a Householder reflection around the vector $\widetilde{\mathbf{y}}=\left[\sqrt{y_1},\dots,\sqrt{y_K}\right]^T$. Such Householder matrices will be also
used in the context of \gls{sdn} reverberators, proposed in the next section, where they will exhibit some sought-after properties, including low computational complexity and desirable normalised echo density profiles.

In order to inject energy in a \gls{dwn}, various methods have been used in the past,
ranging from attaching an additional waveguide where the outgoing wave is ignored~\cite{Smith-III:1985ib} or using 
an adapted impedance such that there is no energy reflected along the outgoing wave, 
to more complex approaches~\cite{evangelista2010player}.
A common approach is to apply an external 
ideal volume velocity source to the junction~\cite{karjalainen2004digital}.
This is equivalent to 
superimposing source pressure, $p_{S}$, to the  pressure  due to the waveguides meeting at the junction,  $p$, thus making the total pressure at the junction
equal to:
\begin{equation}
\overline{p}=p_{S}+p. 
\label{eq:nodepressure}
\end{equation}
In the context of \gls{fdtd} 
models, a source that injects energy in this way is called a \emph{soft-source},
as opposed to a \emph{hard-source}, 
which actively interferes with the propagating
pressure field~\cite{sheaffer,murphy2014source}.
In order to implement equation (\ref{eq:nodepressure}) in the \gls{dwn} structure,
the input from the source can be
distributed uniformly to incoming wave variables according to
\begin{equation}
\overline{\mathbf{p}}^+=\mathbf{p}^++\frac{p_{S}}{2}\mathbf{1}~.
\label{eq:mod}
\end{equation}
which provides the intended node pressure~\cite{murphythesis}. This is the approach that will be used for injecting source energy in the proposed \gls{sdn} structures.

\subsection{Digital Waveguide Meshes}

\glspl{dwn} formed of fine grids 
of scattering junctions, referred to as \emph{ digital waveguide meshes}, are
 used to model wave propagation in an acoustic medium~\cite{bilbao2004wave}.
Each spatial sample in a \gls{dwm} is represented by a
$K$-port scattering junction, connected to its geometric neighbors
over bidirectional unit-sample delay lines. 
In the typical case of an isotropic medium,  $\mathbf{y}$  is a constant vector, while
$\mathbf{A}$ and $\widetilde{\mathbf{A}}$ are identical and given by
\begin{equation}
\mathbf{A}=\frac{2}{K} \mathbf{1}\mathbf{1}^T-\mathbf{I}~.
\end{equation}
We will refer to such a scattering matrix as the \emph{isotropic scattering matrix}.

In the one-dimensional band-limited case, the \gls{dwm} model
provides the exact solution of the wave equation~\cite{smith2010physical}.
In the two~\cite{kelloniemi2006simulation}
and three~\cite{Hacihabiboglu:2009gd} dimensional cases,
sound propagates in a DWM at slightly 
different speeds in different directions and different frequencies,
causing a \emph{dispersion error}~\cite{Strikwerda:2004bs},
which can be controlled and reduced to some extent by 
means of careful design of the mesh topology
or by using interpolated meshes and frequency warping 
methods~\cite{savioja2004interpolated, murphy2007acoustic}.

Accurate modeling with \glspl{dwm} requires mesh topologies with a very fine resolution 
(e.g. $\approx 10^7$ junctions for
a room of size $4\times 6\times 3$~m~\cite{huangdigital}).
That makes the computational load and the amount of memory required prohibitively high for
real-time operation, especially for large rooms.
These drawbacks motivated the work of 
Karjalainen \emph{et al.} reported in~\cite{huangdigital},
which is reviewed in the next subsection. 

\subsection{Reduced Digital Waveguide Meshes}
\begin{figure}[t]
\centering
\includegraphics[width=0.32\textwidth]{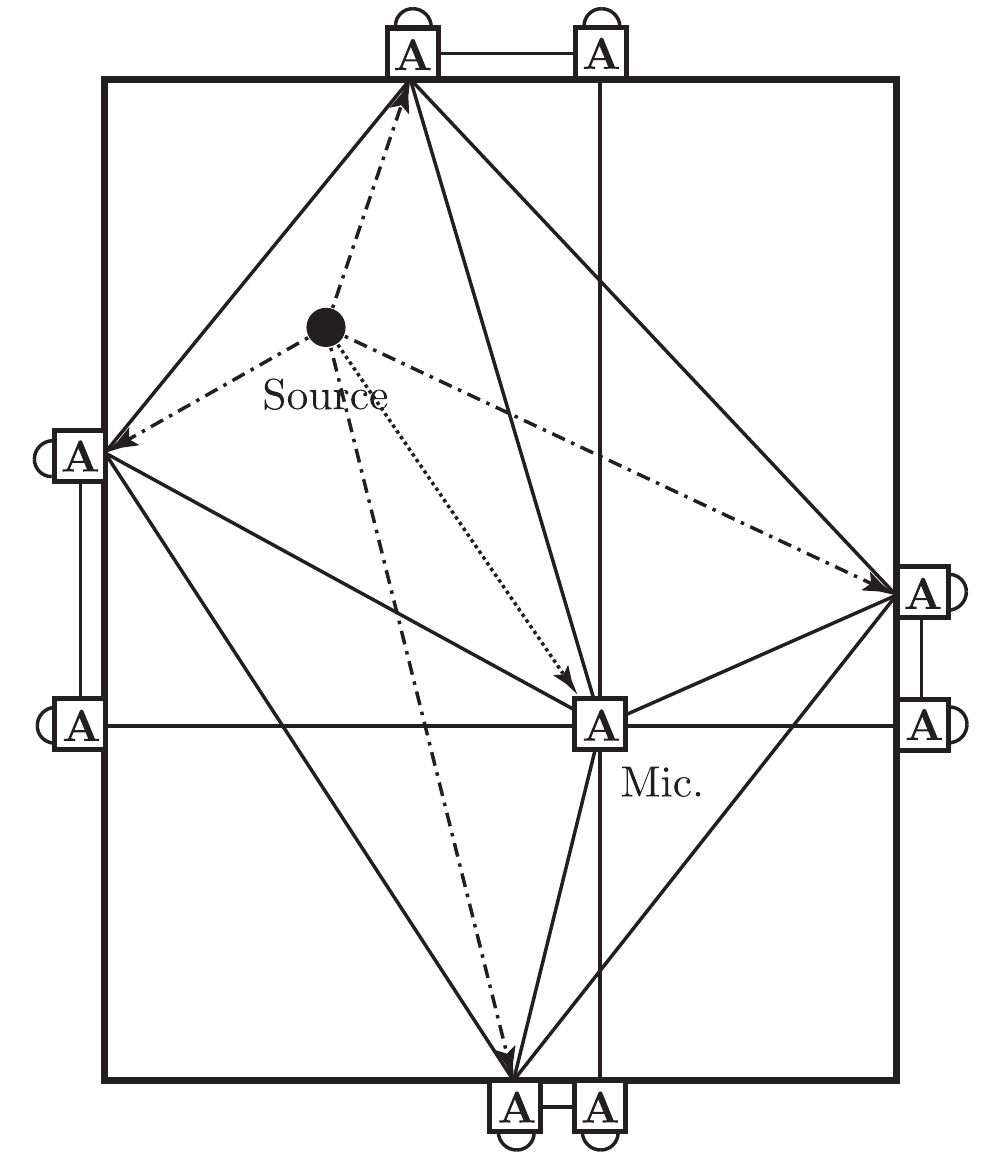}
\caption{Conceptual depiction of one of the \protect\gls{dwn} 
topologies proposed by Karjalainen~et al.\
as seen by an observer above the simulated enclosure 
(in this case a 2D rectangular room)~\cite{huangdigital}.
The solid black lines denote the bidirectional delay lines interconnecting scattering nodes.
The scattering nodes are denoted by the $\mathbf{A}$ blocks,
where $\mathbf{A}$ is the lossless scattering matrix.
The dash-dotted lines denote the unidirectional absorptive delay lines.
The dotted line denotes the line-of-sight (LOS) component.
The solid arcs around junctions denote loaded self-connections implementing losses.
Please note that while this figure represents the case of a 2D rectangular room, all the simulations in this paper use 3D rectangular rooms.}
\label{fig:roompos}
\end{figure}

In order to lower the computational complexity of \gls{dwm} models, 
Karjalainen \emph{et al.} 
considered coarse approximations of room response
synthesis via sparse \gls{dwm}
structures~\cite{huangdigital}.
One such structure is shown in 
\figurename~\ref{fig:roompos}. In this network, the sound source is
connected via unidirectional absorbing delay lines to scattering
junctions. These junctions are positioned at the locations where first-order reflections
impinge on the walls.  This ensures that delays of first-order
reflections are rendered accurately. 
The junctions are connected
via bidirectional delay lines with the microphone, which is also modeled as
a scattering junction contributing to the energy circulation in
the network.  The line-of-sight component is modeled by a direct
connection between the source and the microphone.  Additional
bidirectional delay lines parallel to the wall edges are included to better
simulate room axial modes~\cite{Kuttruff:2000zl}.  All the wall
junctions are connected in a ring topology.

In order to model losses at the walls, it appears that 
a combination of junction loads and additional self-connections is used.
However, implementation details are not given,  
and the authors state that the network 
required careful heuristic tuning~\cite{huangdigital}.
For these reasons the results are difficult to replicate.

In the next section, we describe a 
structure which renders the direct path and first-order early reflections 
accurately both in time and amplitude, 
while producing progressively coarser approximations of higher order reflections and late reverberation.
The proposed method has parameters that are inherited directly 
from room geometry and absorption properties of wall materials,
and thus does not require tuning. In order to distinguish the proposed structure 
from the ones considered by Karjalainen \emph{et al.} in~\cite{huangdigital}, 
and because of the importance of the scattering operation,
we refer to the proposed structures as \emph{scattering delay networks} (\glspl{sdn}).

\section{Scattering Delay Networks}\label{sec:dwmfdn}

\begin{figure}[t]
  \centering
  \includegraphics[width=0.32\textwidth]{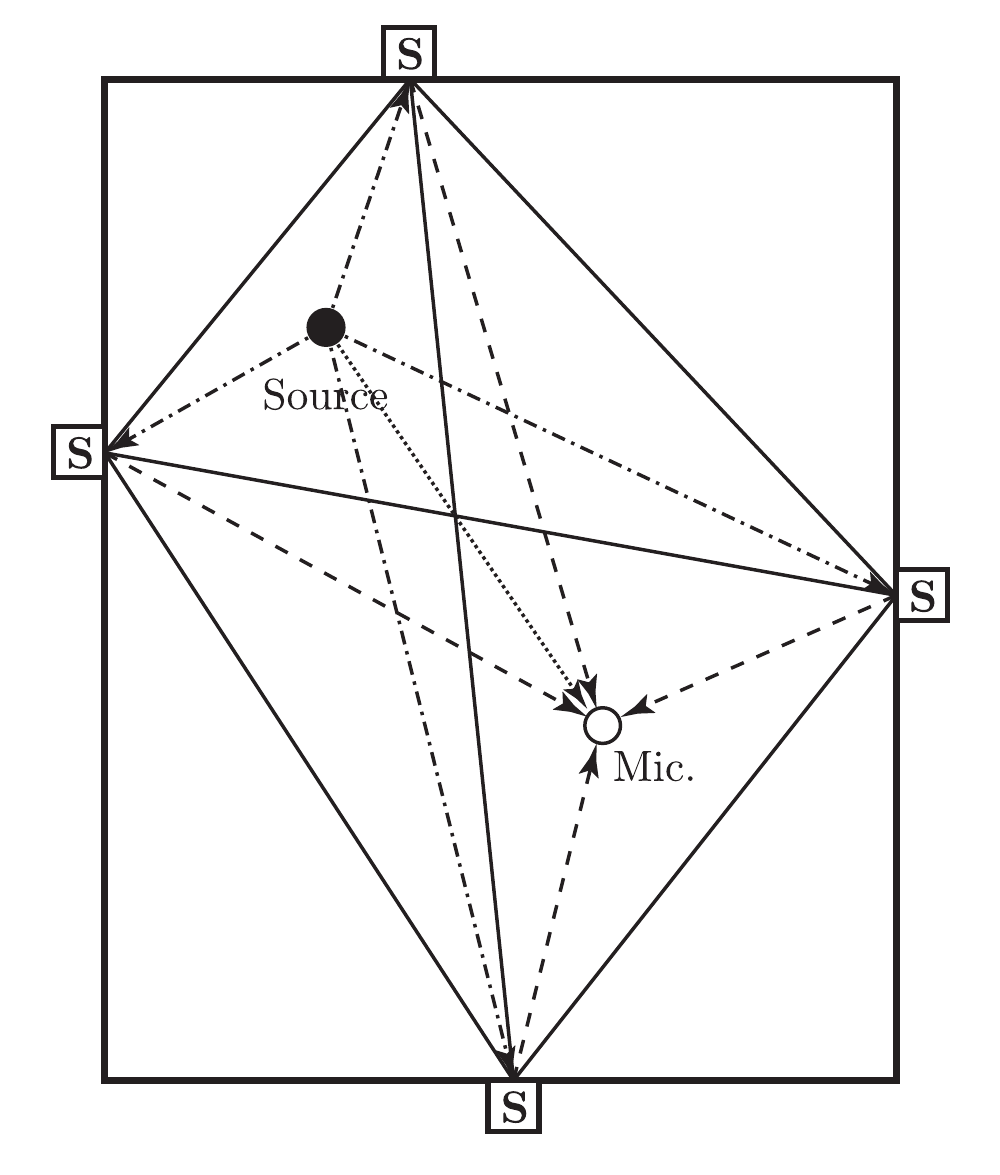}
  \caption{Conceptual depiction of the 
    \protect\gls{sdn} reverberator.
    The solid black lines denote bidirectional delay lines
    interconnecting the \protect\gls{sdn} wall nodes. 
The \protect\gls{sdn} wall nodes are denoted by the $\mathbf{S}$ blocks,
where $\mathbf{S}$ is the lossy scattering matrix. The dash-dotted lines
    denote unidirectional absorptive delay lines connecting the source
    to the \protect\gls{sdn} nodes.  The dashed lines denote unidirectional
    absorptive delay lines connecting the \protect\gls{sdn} nodes to the
    microphone.  The dotted line denotes the direct-path component.}
  \label{fig:SDN}
\end{figure}

The aim of the \gls{sdn} structure is to simulate the acoustics of an enclosure 
using a minimal topology which would ensure that each
significant reflection in the given space has a corresponding
reflection in the synthesized response. This requires representing
each significant reflective surface using one node in a fully
connected mesh topology.  The concept is illustrated in
\figurename\ \ref{fig:SDN}.
For clarity, considerations
 in this paper will pertain to rectangular empty rooms,
however all the presented concepts can be extended in a straightforward
manner to arbitrary polyhedral spaces with additional reflective
surfaces in their interior.  

As shown in
\figurename\ \ref{fig:SDN}, the network consists of a fully connected
\gls{dwn} with one scattering node for each wall. 
This network is minimal in the
sense that the removal of any of the nodes or paths would make it
impossible to model a significant subset of reflections which would
arise in the space.  

The source is connected to the scattering nodes
via unidirectional absorbing lines. 
As opposed to the reverberators proposed 
in~\cite{huangdigital}, 
the microphone node is a passive 
element that does not participate in the energy recirculation, hence
scattering nodes are connected to the microphone via
unidirectional absorbing lines and no energy is injected from the microphone node back to the network.

Early reflections are known to strongly contribute to one's perception
of the size and spaciousness of a room~\cite{rumsey2001spatial}.
For this reason, nodes are positioned on the walls at locations of first-order
reflections. 
Delays and attenuation of the lines connecting the
nodes, source, and microphone are set so that first-order reflections
are rendered accurately in their timing and energy. 

Second-order
reflections are approximated by corresponding paths within the
network. This is illustrated in \figurename\ \ref{fig:refl}.  It can be
observed from the figure that the accuracy of second-order reflections
depends on the particular reflection, but nonetheless the delays of the approximating paths are similar to the actual ones.  As
the reflection order increases, coarser approximations are made. 
Thus, the proposed network behaves equivalently 
to geometric-acoustic methods up to the first-order reflections, 
while providing a gracefully degrading approximation for higher-order ones.

\begin{figure}
\centering
\subfloat[]{
\includegraphics[width=0.25\textwidth]{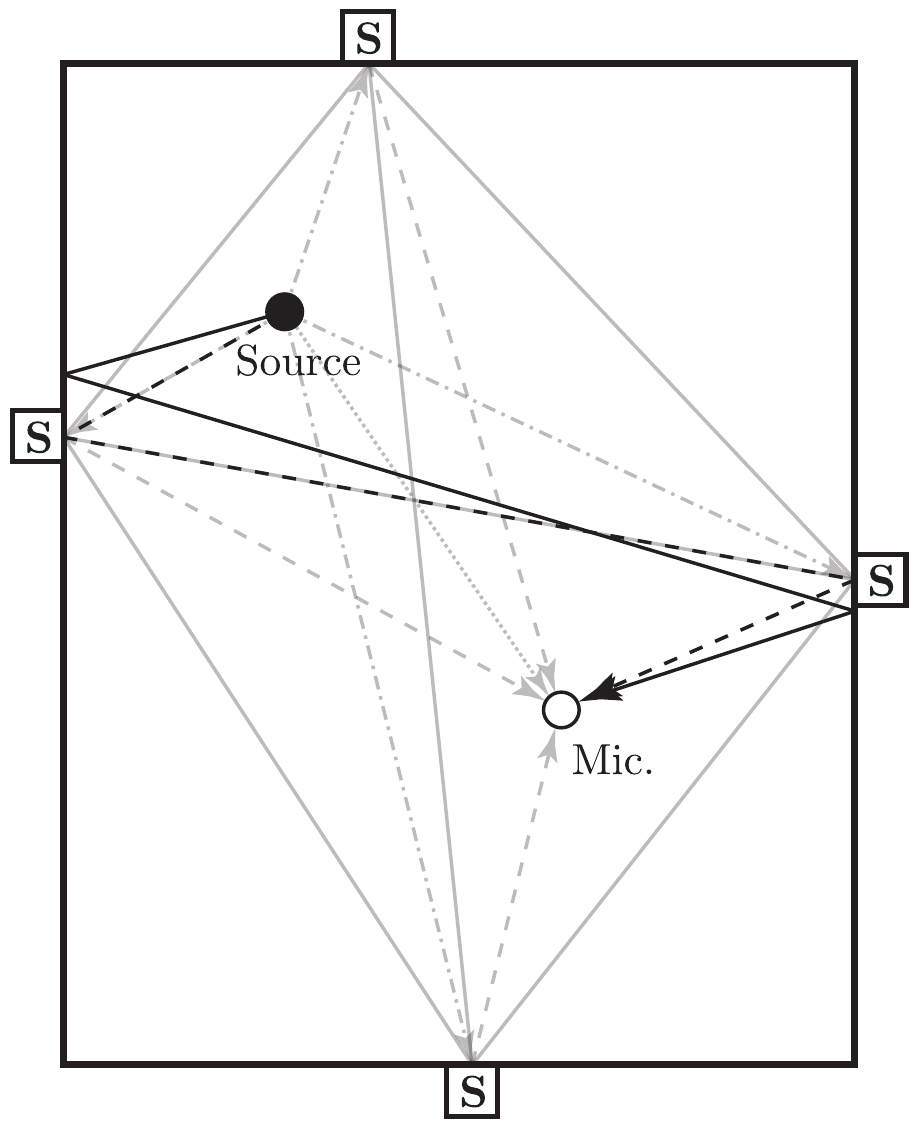}\label{fig:refl_1}}
\subfloat[]{
\includegraphics[width=0.25\textwidth]{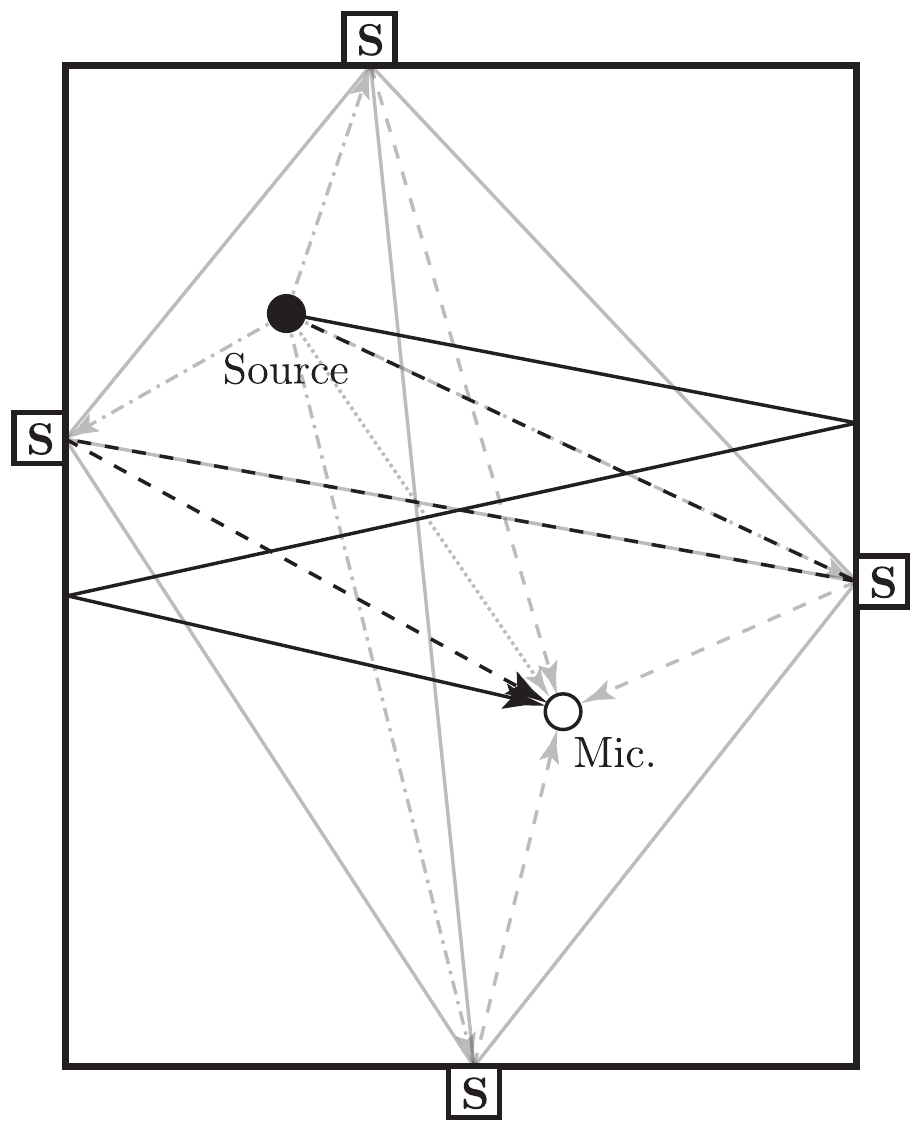}\label{fig:refl_2}}

\subfloat[]{
\includegraphics[width=0.25\textwidth]{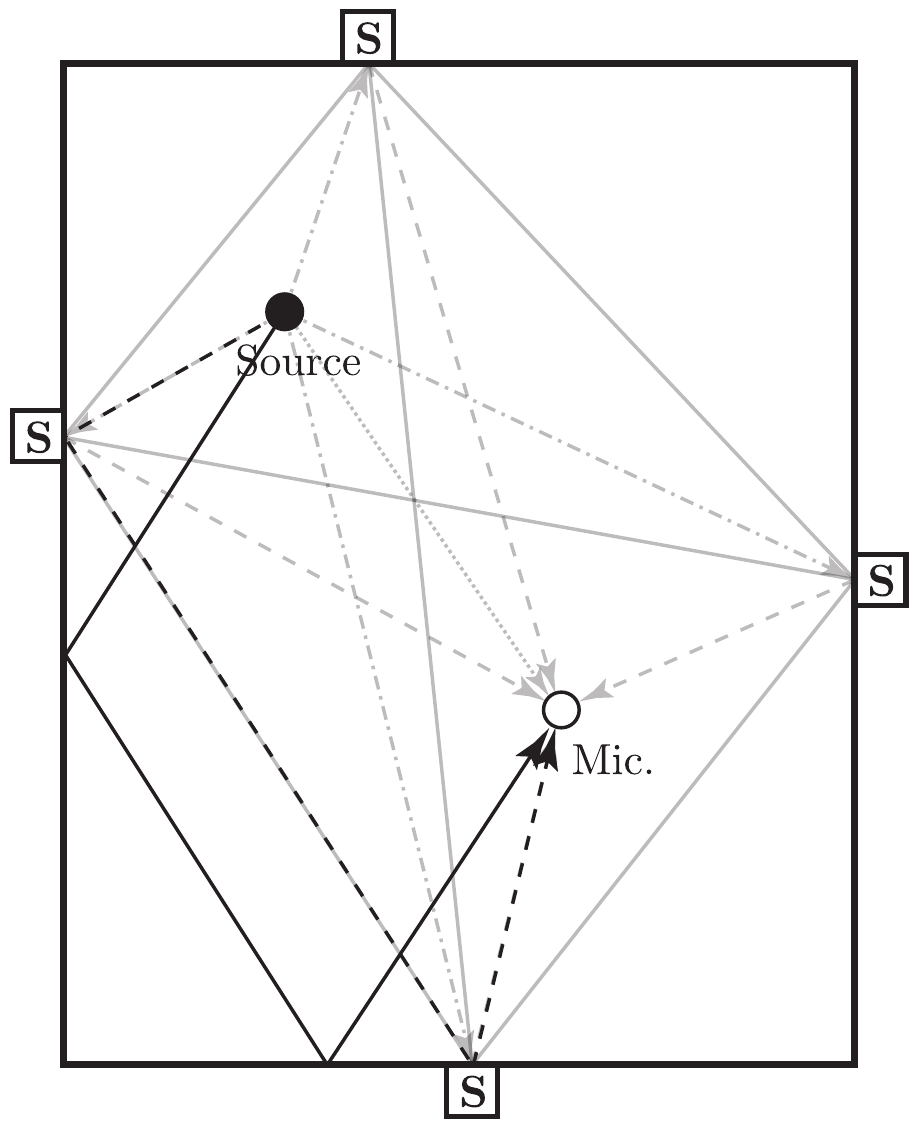}\label{fig:refl_3}}
\subfloat[]{
\includegraphics[width=0.25\textwidth]{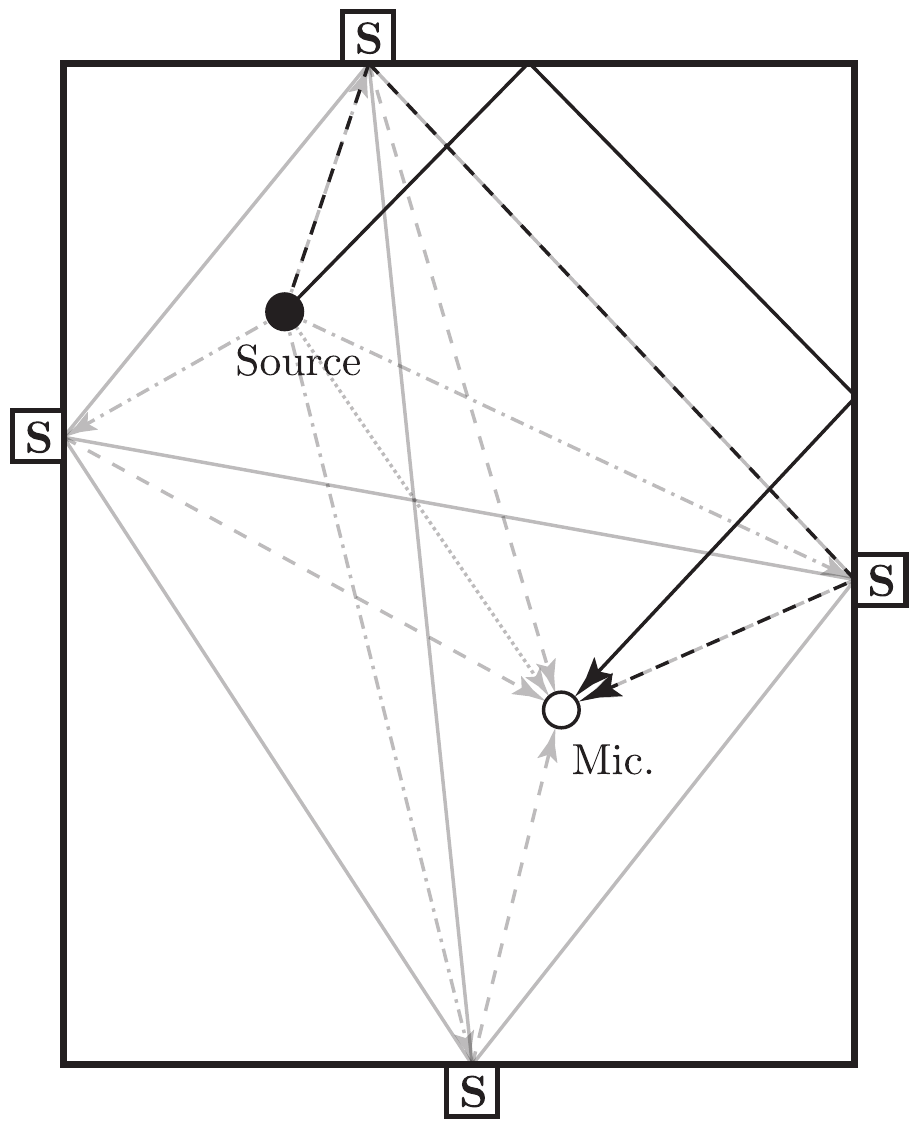}\label{fig:refl_4}}
\caption{Examples of approximations generated by \protect\gls{sdn} 
for four second-order reflections.
The solid black line represents the actual path of the second-order reflection, 
while the dashed line is the corresponding path within the \protect\gls{sdn}.}
\label{fig:refl}
\end{figure}

Precise details of the \gls{sdn} design are given below.

\textbf{Scattering nodes:} Each node is positioned on a wall of the 
modeled enclosure.  The nodes are positioned at the locations where
the first-order reflections impinge on the walls.  These locations are
straightforward to calculate for simple geometries, e.g.\ convex polyhedra.  
The nodes carry out a scattering operation on the inputs
from the other $K$ nodes, $\mathbf{p}^+$, to obtain the outputs 
as $\mathbf{p}^-[n]=\mathbf{S}\mathbf{p}^+[n]$,
where $\mathbf{S}$ is the $K\times K$ (not necessarily lossless) scattering matrix.  
Rectangular rooms, which are used in the following, correspond to $K=5$.
Other geometries where $K$ is a power or $2$ 
are computationally convenient since it can be shown that they  
lead to a multiplier-free realization~\cite{Smith-III:1985ib}.

The scattering matrix $\mathbf{S}$ governs how energy circulates
within the network. 
Since each node is associated with a wall,
it describes how energy is exchanged between walls. 
Furthermore, when incident waves $\mathbf{p}^+$ are scattered from the wall, 
a certain amount of energy is absorbed. 
A macroscopic quantity that describes material absorption 
sufficiently well for most synthesis applications is 
the random-incidence absorption coefficient~\cite{vorlander2008auralization}.
This coefficient, specified by  the  ISO 354 standard,
is known for a variety of materials~\cite{vorlander2008auralization}. 

The wall absorption effect can be expressed in the most general form as
$\mathbf{p}^{-*}\mathbf{Y}\mathbf{p}^-=(1-\alpha)\mathbf{p}^{+*}\mathbf{Y}\mathbf{p}^+$, where 
$\alpha$ is the wall absorption coefficient,
which is equivalent to
\begin{equation}
\mathbf{S}^*\mathbf{Y}\mathbf{S}=(1-\alpha)\mathbf{Y}~.
\label{eq:energyconstr}
\end{equation}
By expressing $\mathbf{S}$ as $\mathbf{S} = \beta\mathbf{A}$, where $\beta=\sqrt{1-\alpha}$, 
the relationship in~(\ref{eq:energyconstr}) becomes equivalent  to 
$\mathbf{A}^*\mathbf{Y}\mathbf{A}=\mathbf{Y}$, i.e. 
the scattering matrix $\mathbf{S}$ 
is the product of the wall reflection coefficient $\beta$
and a lossless scattering matrix $\mathbf{A}$.
As in \glspl{dwn}, if $\mathbf{A}$ is selected as given by (\ref{matrix:unnorm}) 
or (\ref{eq:house}), the propagating variables have a physical interpretation
as the pressure or root-power waves in a network of acoustic tubes.
Other non-physical choices of $\mathbf{A}$ are possible, 
as long as the lossless condition is satisfied; 
these are discussed  in Section~\ref{sec:smselection}.

As will be shown in Section~\ref{sec:rt}, 
setting the scattering matrix as $\mathbf{S} = \beta\mathbf{A}$
results in a rate of energy decay that is consistent with the well-known Sabine and Eyring formulas 
as well as with the results of the \gls{ism}. 
This can be attributed to the fact that both  
the average time between successive reflections (i.e. the \emph{mean free path}) 
and the energy loss at the walls in the virtual \gls{sdn} network 
are close to the ones observed in the corresponding physical room.

The absorption characteristic of most real materials is frequency-dependent.
This can be modeled by using a more general scattering 
matrix $\mathbf{S}(z) = \mathbf{H}(z)\mathbf{A}$, 
where $\mathbf{H}(z)=\text{diag}\{H(z),\dots,H(z)\}$, and $H(z)$ is a wall filter. 
The absorption coefficients  in consecutive octave bands, 
for a range of materials, are reported in~\cite{vorlander2008auralization}.
Standard filter design techniques can be used to fit 
the frequency response $H(e^{j \omega})$ to these tabulated values.
Minimum-phase IIR 
filters are particularly convenient in this context as they reduce computational
load without significantly affecting the phase  of  simulated
reflections~\cite{huopaniemi2002modeling}.

As in conventional~\glspl{dwm}, the pressure at the \gls{sdn} node is
a function of the incoming wave variables, $p_{i}^+[n]$, from 
neighboring nodes and the pressure, $p_{S}[n]$, injected by the source. 
That is modeled according to~(\ref{eq:nodepressure}),
and it is illustrated in \figurename\ \ref{fig:sourcenode}. Other details of source-to-node connections are discussed next.

 \begin{figure}[t]
\centering
\includegraphics[width=0.47\textwidth]{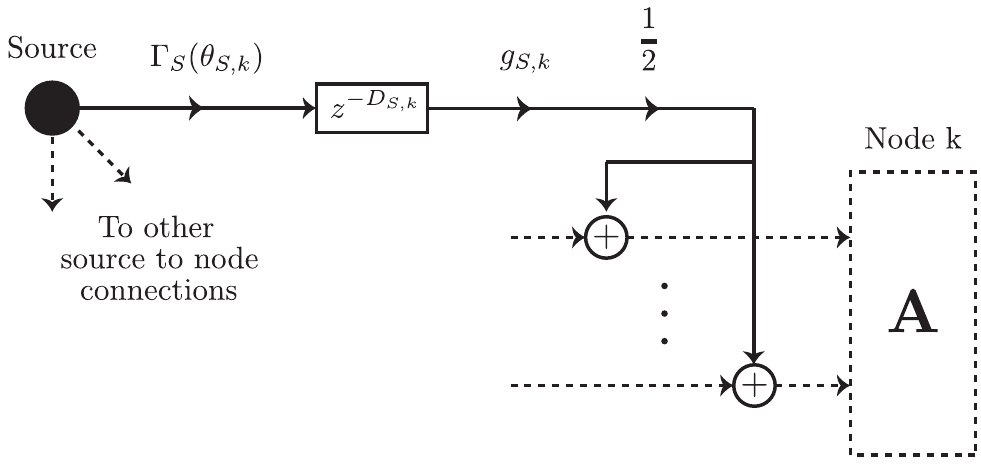}
\caption{Connection between the source node and an \protect\gls{sdn} node.}
\label{fig:sourcenode}
\end{figure}

\textbf{Source-to-node connection:} The input to the system is
provided by a source node connected to \gls{sdn} nodes via
unidirectional absorbing delay lines (see
\figurename\ \ref{fig:sourcenode}).

The delay of the line connecting the source at $\mathbf{x}_S$ and the
\gls{sdn} node positioned at $\mathbf{x}_k$ is given by the corresponding
propagation delay $D_{S,k} =
\lfloor{}F_s\|\mathbf{x}_S-\mathbf{x}_k\|/c\rfloor$, where $c$ is the
speed of sound and $F_s$ is the sampling frequency.
The attenuation due to spherical spreading ($1/r$ law) is modeled as
\begin{equation}
g_{S,k}=\frac{1}{\|\mathbf{x}_S-\mathbf{x}_k\|}.
\label{eq:gsk}
\end{equation}

Source directivity is another important simulation parameter in room acoustic synthesis. 
The sparse sampling of the simulated enclosure prohibits the simulation of
source directivity in detail.  However, a coarse approximation
can be incorporated by weighting the outgoing signals by
$\Gamma_S(\theta_{S,k})$, where $\Gamma_S(\theta)$ is the source
directivity, and $\theta_{S,k}$ is the angle between the source
reference axis and the line connecting the source and $k$-th node.
An alternative approach consists of using an average of the 
directivity pattern in some angular sector. It should be noted that it is possible to simulate frequency-dependent characteristics of source directivity using short, variable linear or minimum-phase filters. However, in order to keep the exposition in this article clear it is assumed that the directivity patterns are independent of frequency and can be modeled as simple gains.

\textbf{Node-to-node connection:} The connections between the \gls{sdn}
nodes consist of bidirectional delay lines modeling the propagation
path delay as shown in \figurename\ \ref{fig:internode}.
Additional low-pass filters can be inserted into the network at this point to model 
the frequency-dependent characteristic of air absorption, 
as proposed by Moorer  in the context of FDNs \cite{moorer1979}.

\begin{figure}[t]
\centering
\includegraphics[width=0.47\textwidth]{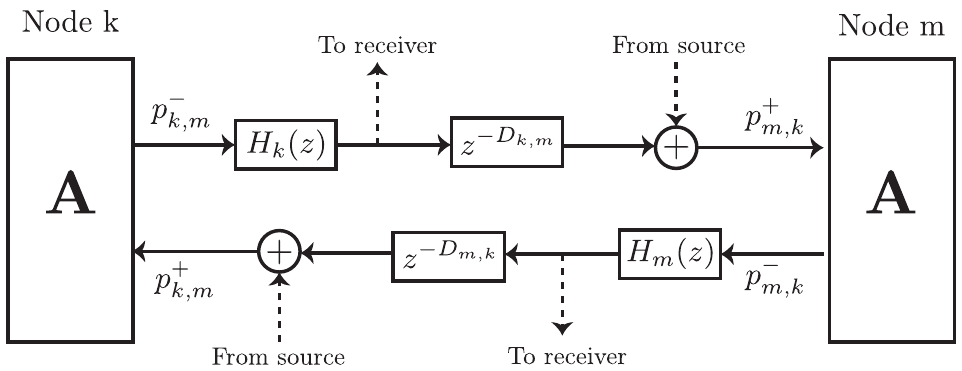}
\caption{Two interconnected SDN nodes.}
\label{fig:internode}
\end{figure}

The delays of the lines connecting the nodes are determined by their
spatial coordinates.
Thus, the delay of the line between a node at location $\mathbf{x}_k$
and a node at $\mathbf{x}_m$ is $D_{k,m}
=\lfloor{}F_s\|\mathbf{x}_k-\mathbf{x}_m\|/c\rfloor$.

\textbf{Node-to-microphone connection:} Each node is
connected to the microphone node via a unidirectional absorbing delay line.  
The signal extracted from the junction, $p_e[n]$, 
is a linear combination of outgoing wave
variables, $p_e[n]=\mathbf{w}^T\mathbf{p}^-[n]$, 
where $\mathbf{p}^-[n]$ is the wave vector after the wall filtering operation, 
as shown in \figurename~\ref{fig:nodemic}.

In the physical case where $\mathbf{A}$ 
is in the form~(\ref{matrix:unnorm}) or~(\ref{eq:house}), 
the outgoing signal to the microphone is taken as 
the node's pressure, as given by equations (\ref{eq:pressurep}) 
and (\ref{eq:nodepressure}):
\begin{equation}
p_e[n]=\frac{2}{\langle \mathbf{1}, \mathbf{y}\rangle}\mathbf{y}^T \mathbf{p}^-[n]~.
\label{eq:weightphysical}
\end{equation}
In the non-physical case,
various choices are available for extracting a signal from the junction.
The only condition that  the weights $\mathbf{w}$ need to satisfy
is that the cascade of pressure injection, scattering and extraction
does not alter the amplitude of first-order reflections.
Since the incoming wave vector $\mathbf{p}^+$
for a first-order reflection with amplitude $p_S$ is given by 
$\mathbf{p}^+=(\frac{p_S}{2}\mathbf{1})$ (see \figurename~\ref{fig:sourcenode})
and since $\mathbf{p}^-=\mathbf{S}\mathbf{p}^+$, 
this condition can be written as 
$\mathbf{w}^T\mathbf{S}\left(\frac{p_S}{2}\mathbf{1}\right)=\beta p_S$ or, equivalently, as
\begin{equation}
\mathbf{w}^T\mathbf{A}\mathbf{1}=2~.
\label{eq:condweights}
\end{equation}

A possible choice for $\mathbf{w}$ is a constant vector,
which is computationally convenient since it requires a single multiplication.
In this case, the constraint~(\ref{eq:condweights}) yields 
the unique solution~$\mathbf{w}=\frac{2}{\mathbf{1}^T\mathbf{A}\mathbf{1}}\mathbf{1}$.

The delay from the $k$-th \gls{sdn} node to the microphone node is
$D_{k,M} = \lfloor{}F_s\|\mathbf{x}_k-\mathbf{x}_M\|/c\rfloor$.  As
with the source directivity, the microphone directivity is also modeled
using a simple gain element $\Gamma_M(\theta_{k,M})$, where
$\Gamma_M(\theta)$ is the microphone directivity pattern and
$\theta_{k,M}$ is the angle between the microphone acoustical axis and
the $k$-th node.  This approach ensures that the microphone is
emulated correctly for the directions associated to the first-order
reflections. As with the source-node connections, the microphone directivity can also be modeled using short, variable linear or minimum-phase filters. However, simple gain elements are preferred in this article for clarity of presentation. Similarly, as with the source-node connection, an alternative approach consists of using an average of the directivity pattern in some angular sector.

\begin{figure}[t]
\centering \includegraphics[width=0.47\textwidth]{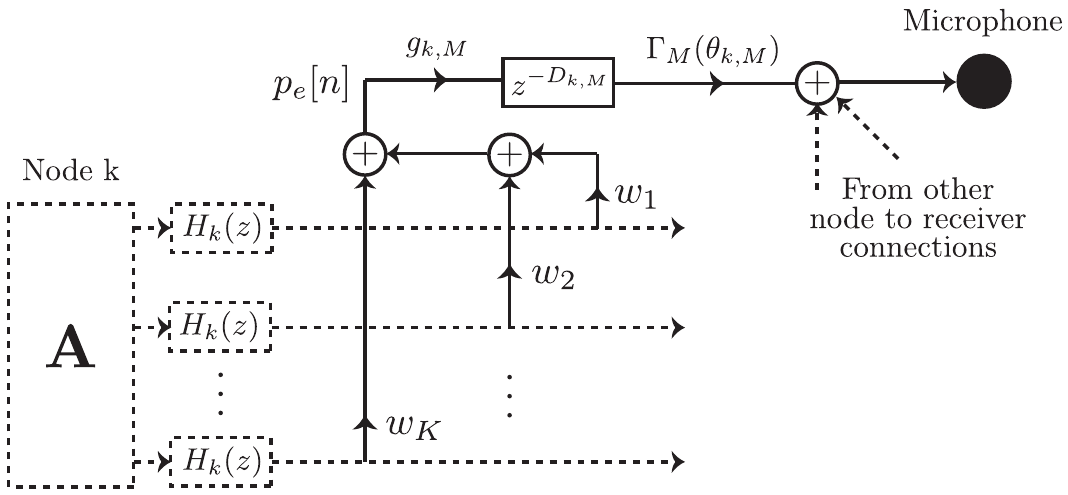}
\caption{Connection between an \protect\gls{sdn} node and the
  receiver/microphone node.}
\label{fig:nodemic}
\end{figure}

The attenuation coefficient is set as
\begin{equation}
g_{k,M}=\frac{1}{1+\frac{\|\mathbf{x}_k-\mathbf{x}_M\|}{\|\mathbf{x}_S-\mathbf{x}_k\|}},
\end{equation}
such that, using (\ref{eq:gsk}),
\begin{equation}
g_{S,k}\times
g_{k,M}=\frac{1}{\|\mathbf{x}_S-\mathbf{x}_k\|+\|\mathbf{x}_k-\mathbf{x}_M\|},
\label{eq:reflectmultiplication}
\end{equation}
which yields the correct attenuation for the first-order reflection
according to the $1/r$ law.  Notice that the above choice of $g_{k,M}$
and $g_{S,k}$ is not unique in satisfying the constraint
(\ref{eq:reflectmultiplication}).  The attenuation can in fact be
distributed differently between the source-to-node and
node-to-microphone branches 
but with little impact on overall energy decay.

\section{Properties of SDNs}
\label{sec:propsdn}

\subsection{Transfer function}
\label{sec:transffunc}

\begin{figure*}
\centering \includegraphics[width=0.9\textwidth]{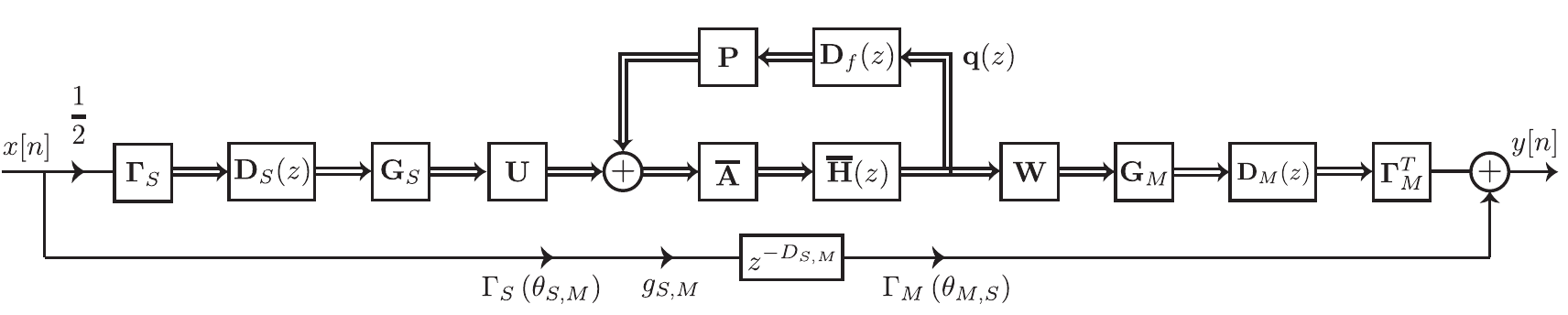}
\caption{Block diagram of the \protect\gls{sdn} reverberator.}
\label{fig:sdn}
\end{figure*}

The block diagram of an \gls{sdn} system is shown in
\figurename\ \ref{fig:sdn}. In the figure,
\begin{align}
\vec{\Gamma}_S&=\left[\Gamma_{S}\left(\theta_{S,1}\right),
\Gamma_{S}\left(\theta_{S,2}\right),
\dots,\Gamma_{S}\left(\theta_{S,K+1}\right)\right]^T~,
\label{eq:source_dir} \\
\mathbf{D}_S(z)&=\textrm{diag}\left\{z^{-D_{S,1}},
z^{-D_{S,2}},
\dots,z^{-D_{S,K+1}}\right\}~,
\label{eq:source_delay}\\
\mathbf{G}_S&=\textrm{diag}\left\{g_{S,1},
g_{S,2},
\dots,g_{S,K+1}\right\}~,
\label{eq:source_gain}
\end{align}
are  the source directivity vector, the source delay matrix, and
the source attenuation matrix, respectively.  The corresponding quantities
associated with the microphone $\vec{\Gamma}_M, \mathbf{D}_M(z)$ and
$\mathbf{G}_M$ are defined as in (\ref{eq:source_dir}),
(\ref{eq:source_delay}) and (\ref{eq:source_gain}), by substituting
$S$ with $M$.  
Further, 
\begin{align}
\mathbf{U}&=\textrm{diag}\Big\{\underbrace{\mathbf{1},\dots,\mathbf{1}}_{K+1}\Big\}~,\\
\mathbf{D}_f(z)&=\textrm{diag}\left\{z^{-D_{1,2}},\dots,z^{-D_{K,{K+1}}}\right\}~, \\
\overline{\mathbf{H}}(z)&=\textrm{diag}\Big\{\underbrace{H_{1}(z),\dots,H_{1}(z)}_{K},
\dots,H_{K+1}(z)\Big\}~,
\end{align}
are the block-diagonal matrix that distributes source signals to input wave variables,
the inter-node delay matrix,
and the wall absorption matrix, respectively.
The block diagonal scattering matrix $\overline{\mathbf{A}}$ and 
the block diagonal weight matrix $\mathbf{W}$ are defined as 
\begin{align}
\overline{\mathbf{A}}&=\textrm{diag}\left\{\mathbf{A}_1,\mathbf{A}_2,\dots\ ,\mathbf{A}_{K+1}\right\}~,\\
\mathbf{W}&=\textrm{diag}\left\{\mathbf{w}^T_1,\dots,\mathbf{w}^T_{K+1}\right\}~,
\end{align}
where $\mathbf{A}_k$ and $\mathbf{w}_k$ are the $i$-th node's scattering matrix 
and pressure extraction weights, respectively.
While these variables can in general be different for each node, 
in all the simulations presented in this paper they are selected to be equal, 
$\mathbf{A}_1=\dots=\mathbf{A}_{K+1}$
and $\mathbf{w}_1=\dots=\mathbf{w}_{K+1}$.

Finally, in \figurename\ \ref{fig:sdn}, $g_{S,M}$ and
$z^{-D_{S,M}}$ are the line-of-sight attenuation and delay,
respectively, 
and $\mathbf {P}$ is a permutation matrix whose elements are
determined based on the network topology.  Due to the
underlying node connectivity being bidirectional, this permutation
matrix is symmetric.  For the simplest case of a three-dimensional
enclosure with a rectangular shape, the permutation matrix takes
the form $\mathbf {P} = \delta_{i,f (j)}$, where $\delta_{i,j}$ is the
Kronecker delta, 
\begin{equation}
f(i) = ((6i-((i-1))_N-1))_{N(N-1)} + 1~,
\end{equation}
and $((\cdot))_N$ is the modulo-$N$ operation.
Inspection of \figurename~\ref{fig:sdn} reveals
that the system output can be expressed as
\begin{equation}
Y(z)=\mathbf{\Gamma}_M^T \mathbf{D}_M(z) \mathbf{G}_M \mathbf{W}
\mathbf{q}(z)+\overline{g}z^{-D_{S,M}}X(z)
\end{equation}
where
$\overline{g}={g}_{S,M}\Gamma_S\left(\theta_{S,M}\right)\Gamma_M\left(\theta_{M,S}\right)$,
and $\mathbf{q}(z)$ is the state vector, given by
\begin{equation}
\mathbf{q}(z)=\tfrac{1}{2}\left[\mathbf{I}-\overline{\mathbf{H}}(z)
  \mathbf{\overline{A}} \mathbf{P}
  \mathbf{D}_f(z)\right]^{-1}\overline{\mathbf{H}}(z) \mathbf{\overline{A}}
\mathbf{U}\mathbf{G}_S \mathbf{D}_S(z)\vec{\Gamma}_SX(z)~.
\label{eq:qz}
\end{equation}
 The transfer function can therefore be expressed as
\begin{equation}
H(z)=\tfrac{1}{K}\mathbf{k}_M^T(z)\left[\mathbf{\overline{A}}^{T}\overline{\mathbf{H}}^{-1}(z)-\mathbf{P}
  \mathbf{D}_f(z)\right]^{-1}\mathbf{k}_S(z)+\overline{g}z^{-D_{S,M}}~,
\label{eq:transferfunct}
\end{equation}
where $\mathbf{k}_M^T(z)=\mathbf{\Gamma}_M^T
\mathbf{D}_M(z)\mathbf{G}_M \mathbf{W}$ and
$\mathbf{k}_S(z)=\mathbf{U}\mathbf{G}_S\mathbf{D}_S(z)\mathbf{\Gamma}_S$.

It may be observed that, unlike \gls{fdn} reverberators, relevant
acoustical aspects, such as the direct path and reflection delays, are
modeled explicitly,  allowing complete control of 
source and microphone positions and their directivity patterns.

Expressing the system transfer function as given above allows for a
complete demarcation of the directional properties and positions of
the source and microphone, wall absorption characteristics and room
shape and size. This is especially useful in keeping computational
cost low for cases where only a single aspect, such as  source
orientation, changes.

\subsection{Stability}

The stability of the \gls{sdn} follows from the fact that its recursive part, 
i.e. the backbone formed by the \gls{sdn} nodes, is a fully connected~\gls{dwn}. 
The stability of lossless \glspl{dwn} is, in turn, guaranteed by the fact 
that the network has a physical interpretation as a network of acoustic tubes~\cite{smith1997aspects},\cite{Rocchesso:2002rq}.
Indeed, the ideal physical system's total energy provides a Lyapunov function bounding the sum-of-squares in the numerical simulation (provided one uses rounding toward zero or error feedback in the computations)~\cite{smith1997aspects}.
Furthermore, the network conserves its stability when losses due to wall absorption are included at the \gls{sdn} nodes since the physical pressure (i.e. the sum of incoming and outgoing pressure waves) is then always reduced relative to the lossless case.

\subsection{Lossless Scattering Matrices}
\label{sec:smselection}

In this section, we explore possible 
choices for lossless scattering matrices and discuss their implications.
First we present a complete parametrization of real lossless matrices.
The parameterization is an immediate corollary of the following   theorem.
\begin{theorem}
\label{theorem:param}
A real  square matrix $\mathbf{A}$ is diagonalizable  if and only if it has the  form 
\begin{equation}
\mathbf{A}=\mathbf{T}^{-1}\mathbf{\Lambda} \mathbf{T}~,
\label{eq:aform}
\end{equation}
where $\mathbf{T}$ is a real invertible matrix, and
$\mathbf{\Lambda}$ is a block diagonal matrix consisting of  $1\times 1$ blocks which are real eigenvalues of 
$\mathbf{A}$,  and $2\times 2$ blocks of the form 
$$
\arraycolsep=4pt\def\arraystretch{0.4}
  \begin{bmatrix}
    0 & -r_i  \\
    r_i & 2r_i\cos(\theta_i)
  \end{bmatrix}~,$$
  where $r_ie^{j\theta_i}$ are complex eigenvalues of $\mathbf{A}$ which appear in pairs with their conjugates.
 \end{theorem}
\noindent The theorem is proved in the appendix. 

\begin{corollary}\label{cor:par}
A real square matrix is lossless if and only if it has the form given by Theorem \ref{theorem:param} with all eigenvalues on the unit circle.
\end{corollary}
The corollary follows from the fact that a square matrix is lossless if and only if it is diagonalizable and all its eigenvalues are on the unit circle~\cite{Rocchesso:2002rq}.

Within the space of lossless matrices, a large degree of
leeway is left for pursuing various physical or perceptual criteria.
This can be achieved, for instance, 
by finding a lossless matrix~$\mathbf{A}$ which
minimizes the distance from a real matrix~$\mathbf{D}$ that reflects  some 
sought-after physical or perceptual properties. 
The design then amounts to 
constrained optimization:
\begin{equation}
\underset{\mathbf{A}}{\operatorname{min}} ~
\norm{\mathbf{A}-\mathbf{D}}^2_F ~,
\end{equation}
where $\norm{\cdot}_F$ denotes the Frobenius norm, under the constraint that
$\mathbf{A}$ has the form given in Corollary \ref{cor:par}. This most general case may, however,
involve optimization over an excessive number of parameters. If we restrict  the optimization domain to orthogonal matrices,
solutions can be found without the need for numerical procedures. In particular, the following Theorem result holds:

\begin{theorem} 
The solution to the following optimisation problem
\begin{equation}
\underset{\mathbf{A}}{\operatorname{argmin}} ~
\norm{\mathbf{A}-\mathbf{D}}^2_F ~,~\mathbf{A}^T\mathbf{A}=\mathbf{I}~
\end{equation}
is given by  $\mathbf{A}=\mathbf{U}\mathbf{V}^T$,
where $\mathbf{U}$ and $\mathbf{V}$ are respectively the matrices of
left and right singular vectors of $\mathbf{D}$.
\end{theorem}
\noindent A proof of this result can be found in~\cite{zhang2000flexible}.

Orthogonal matrices which are also circulant have the interpretation of performing all-pass circulant convolution
of the incoming variables, which as discussed below, reduces computational complexity of the scattering operator.
Furthermore, if a certain distribution of (unit-norm) eigenvalues is sought,
the associated circulant matrix can be found by means of a single
inverse \gls{fft}~\cite{Rocchesso:2002rq}.
An in-depth study of such scattering matrices in the context of \gls{dwn} reverberators has been presented in \cite{Rocchesso:2002rq}.

Householder reflection matrices, given by $\mathbf{A}=2\vec{v}\vec{v}^T-\mathbf{I}, ~\|\vec{v}\|=1$, are the subclass of
orthogonal scattering matrices commonly used in the context of \gls{dwn} reverberators.
As discussed in Section~\ref{sec:dwn}, they 
enable a physical interpretation of the propagating variables as normalized
pressure waves in a network of acoustic tubes.
The optimization problem that minimizes the distance from a 
targeted scattering matrix in this case can be expressed as  
\begin{equation}
\underset{\vec{v}}{\operatorname{argmin}}~
\norm{(2\vec{v}\vec{v}^T-\mathbf{I})-\mathbf{D}}^2_F
~,~\|\vec{v}\|=1~,
\label{eq:scatproblem}
\end{equation}
the solution of which is given by the following theorem.
\begin{theorem}
The solution to  the optimization problem in (\ref{eq:scatproblem}) 
is the singular vector corresponding to the largest singular value of  matrix
$\mathbf{D}+\mathbf{D}^T$.
\label{theorem:houseeigen}
\end{theorem}
\noindent The theorem is proved in the appendix.

The  isotropic scattering matrix, {\sl i.e.} the particular case of a Householder reflection obtained for 
$\vec{v}= \mathbf{1}/\sqrt{K}$, can be physically interpreted as the scattering matrix 
which 
takes a reflection from one node and distributes its energy equally
among all other nodes. This is the only orthogonal matrix which has this
property, as stated by the following theorem. 
\begin{theorem}
If $\mathbf{A}$ is an orthogonal matrix which scatters the
energy from each incoming direction uniformly among all other
directions, then it must have the form $\mathbf{A} = \pm
\frac{2}{K}\mathbf{1}\mathbf{1}^T-\mathbf{I}$.
\label{theorem:isotropic}
\end{theorem}
\noindent The theorem is proved in the appendix.
The isotropic matrix thus satisfies some optimality criteria and, as discussed below, allows for
fast implementation. We will see in Section~\ref{sec:echo} that its special
structure leads to a fast build up of echo density which is a 
perceptually desirable quality \cite{huang2007aspects}.

These different choices of scattering matrices have their implications on computational complexity.
The computational complexity of the matrix-vector multiplication using  general
lossless and orthogonal matrices is $O[K^2]$ operations. 
Circulant lossless matrices require only 
$O[K\log(K)]$ operations~\cite{Rocchesso:2002rq}.
The computational complexity associated with  matrices
which have the form of Householder reflections is further reduced to $O[K]$.
Among general Householder reflections, which require $2K-1$
additions and $2K$ multiplications, 
the isotropic scattering matrix only requires 
$2K-1$ additions and $1$ multiplication.
Among lossless matrices, 
the ones that require the least operations are permutation matrices.
However, we will see in Section~\ref{sec:echo} that permutation matrices  
lead to an insufficient echo density.

\subsection{Interactivity and multichannel auralization}

Interactive operation of the \gls{sdn} reverberator is accomplished by
updating the model to reflect changes in the positions and rotations of
the source and microphone.  This requires readjusting the positions of the
wall nodes, and updating the delay line lengths and gains accordingly.
It was observed in informal listening tests that updating the delay
line lengths did not cause audible artifacts as long as microphone and
source speeds are within reasonable limits.

The proposed model allows approximate emulation of coincident and non-coincident
recording formats.  Coincident formats
(e.g. Ambisonics~\cite{Gerzon:1985rw},
\gls{hoa}~\cite{Poletti:2000dk}, \gls{vbap}~\cite{Pulkki:1997jt}) can
be easily employed by adjusting the microphone gains
$\Gamma_{M,k}(\theta)$ appropriately.

Non-coincident formats (e.g.~\cite{Johnston:2000oe, desena2013}) 
can be emulated by considering a
separate \gls{sdn}  for each microphone. 
This results in a higher inter-channel decorrelation,
which is what would actually occur in real recordings. However, the overall
computational load also increases.  If 
simulation speed is critical, an alternative approach would be to share the same
set of wall nodes among all the \glspl{sdn}, while creating
dedicated node-to-microphone connection lines for each microphone.

\subsection{Computational load}

This section presents an analysis of the computational complexity of the proposed model
in comparison to the two conceptually closest technologies, \gls{fdn} and the \gls{ism}.
We use the number of floating point operations per second 
as an approximate indicator of the overall computational complexity.
Furthermore, to simplify calculations 
we make the assumption that additions and multiplications
carry the same cost.
This approximation is motivated by 
the progressive convergence of computation time of various operations in 
modern mathematical processing units.

The number of \gls{flops} performed by an \gls{sdn} 
can be calculated as $F_s \left[2K^3+(P+2)K^2+K + 1\right]$,
where $P$ is the number of operations required by each wall filter for each sample. 
Consider now an \gls{fdn} with a $Q\times Q$ feedback matrix.
From inspection of \figurename~\ref{fig:jotblock} 
the computational complexity can be shown to be $F_s \left[2Q^2+(P+3)Q+1\right]$ 
\gls{flops}.
\begin{figure}
\centering
\includegraphics[width=0.47\textwidth]{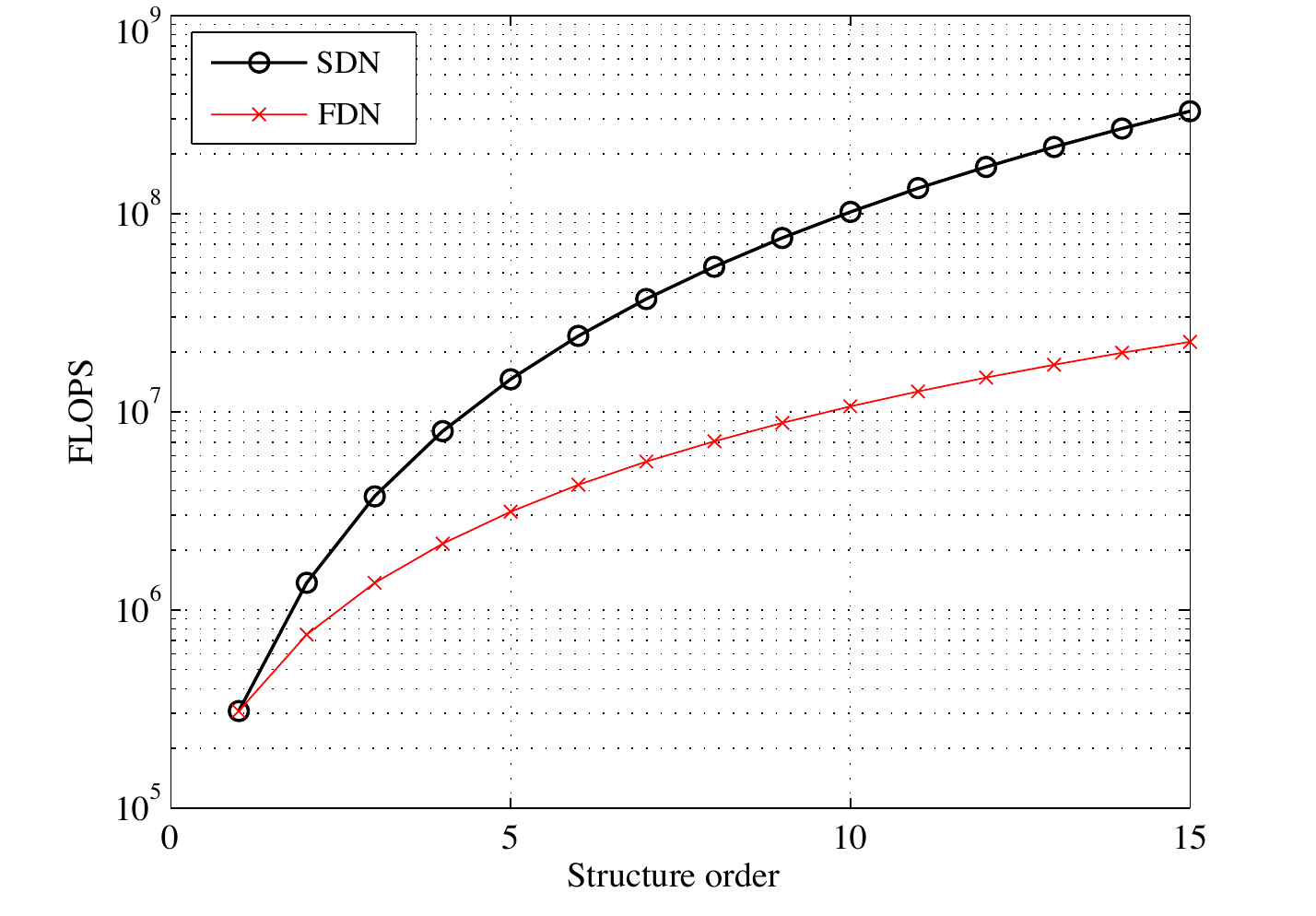}
\caption{
Comparison of computational complexity for \protect\gls{sdn} and \protect\gls{fdn} as a function of the structure order, 
i.e. the size of the feedback matrix  for \protect\gls{fdn}
and the number of neighboring nodes $K$ for \protect\gls{sdn}
(notice that the case of rectangular rooms corresponds to $K=5$).
The sampling frequency is $F_s=44100$ Hz.}
\label{fig:plot_perf_2}
\end{figure}
\figurename~\ref{fig:plot_perf_2} shows a comparison of the computational complexity of 
\glspl{sdn} and \glspl{fdn}.
In this figure, the x-axis denotes the structure order, which we define as 
the number of neighboring nodes $K$ for \gls{sdn} 
and the size of the feedback matrix $Q$ for \gls{fdn}.
The sampling frequency is set to $F_s = 44100$ Hz, 
and the filtering step consists of a simple frequency independent gain ($P=1$)
for both \gls{sdn} and \gls{fdn}.
It may be observed that for the typical case of a 3D rectangular room ($K=5$),
\gls{sdn} has around the same
computational complexity of an \gls{fdn} with a $12\times 12$ feedback matrix.
More specifically, \gls{sdn} for $K=5$ has a complexity of $14.6$ MFLOPS,
while \gls{fdn} for $Q=12$ has a complexity of $14.9$ MFLOPS.
Notice that 
the computational complexity of both \gls{sdn} and \gls{fdn} 
can be reduced significantly by using efficient lossless matrices 
of the type discussed in Section \ref{sec:smselection},
e.g. Householder and circulant lossless matrices.

Consider now the computational cost 
of the \gls{ism} method for rectangular rooms in its original 
time-domain implementation~\cite{Allen79b}.
In the \gls{ism}, each image source requires $10$ floating-point operations\footnote{Here we assume that exponentiations and square roots count as a single floating-point operation.}
to calculate the time index and $1/r$ attenuation of the image,
and $15$ floating-point operations to calculate its 
attenuation due to wall absorption.
This amounts to $25$ floating-point operations in total for each image source.
The number of image sources contained 
in an impulse response of length $T_{60}$ seconds
is approximately equal to the number of room cuboids that fit in
a sphere with radius $cT_{60}$ meters.
This gives a computational complexity of around 
\begin{equation}
25 \left\lceil\frac{\frac{4}{3}\pi(T_{60} c)^3}{L_x L_y L_z}\right\rceil~\text{\gls{flops}},
\label{eq:complism}
\end{equation} 
where $L_x, L_y$ and $L_z$ are the room dimensions.
Here, we are implicitly assuming that the entire \gls{rir} is calculated using the \gls{ism}. 
While this ensures a fair comparison in terms of the other properties assessed in the next section, 
it should be noted that, in order reduce the complexity, most auralization methods
use the \gls{ism} only to simulate the early reflections. 
This is especially so in case of non-rectangular rooms, 
which require a considerably larger amount of memory 
and computational power than equation (\ref{eq:complism})~\cite{borish1984}. 
The rest of the \gls{rir} 
is usually generated using lightweight statistical methods,
e.g. random noise with appropriate energy decay~\cite{lehmann2010diffuse}.

Once the \gls{rir} has been generated, 
it has to be convolved with the input signal 
to obtain the reverberant signal. 
In the case of real-time applications, this can be done efficiently 
using the overlap-add method~\cite{oppenheim1989discrete}. 
The method calculates $2$ \glspl{fft} of length $N$, 
$N$ complex multiplications, $1$ inverse \gls{fft},
and $\left\lceil T_{60}F_s\right\rceil$ real additions for each time frame.
In order for the circular convolution to be equal to the linear convolution,
$N$ must satisfy 
$N\geq \left\lceil{F_s}/{F_r}\right\rceil+\left\lceil T_{60}F_s\right\rceil-1$,
where $F_r$ is the frame refresh rate.
In the best-case scenario where $N$ is a power of $2$, 
the asymptotic cost of each \gls{fft} is $6N\log_2N$~\cite{oppenheim1989discrete}.
Furthermore, each complex multiplication requires $4$ real multiplications and $2$ additions.
Overall, the computational complexity of the overlap-add method is 
$F_r(18N\log_2 N+6N +\left\lceil T_{60}F_s\right\rceil-1)$~\gls{flops}.
In the static case,
the \gls{fft} of the impulse response can be precomputed,
and the cost reduces to
$F_r(12N\log_2 N+6N +\left\lceil T_{60}F_s\right\rceil-1)$~\gls{flops}. 

\begin{figure}
\centering
\includegraphics[width=0.47\textwidth]{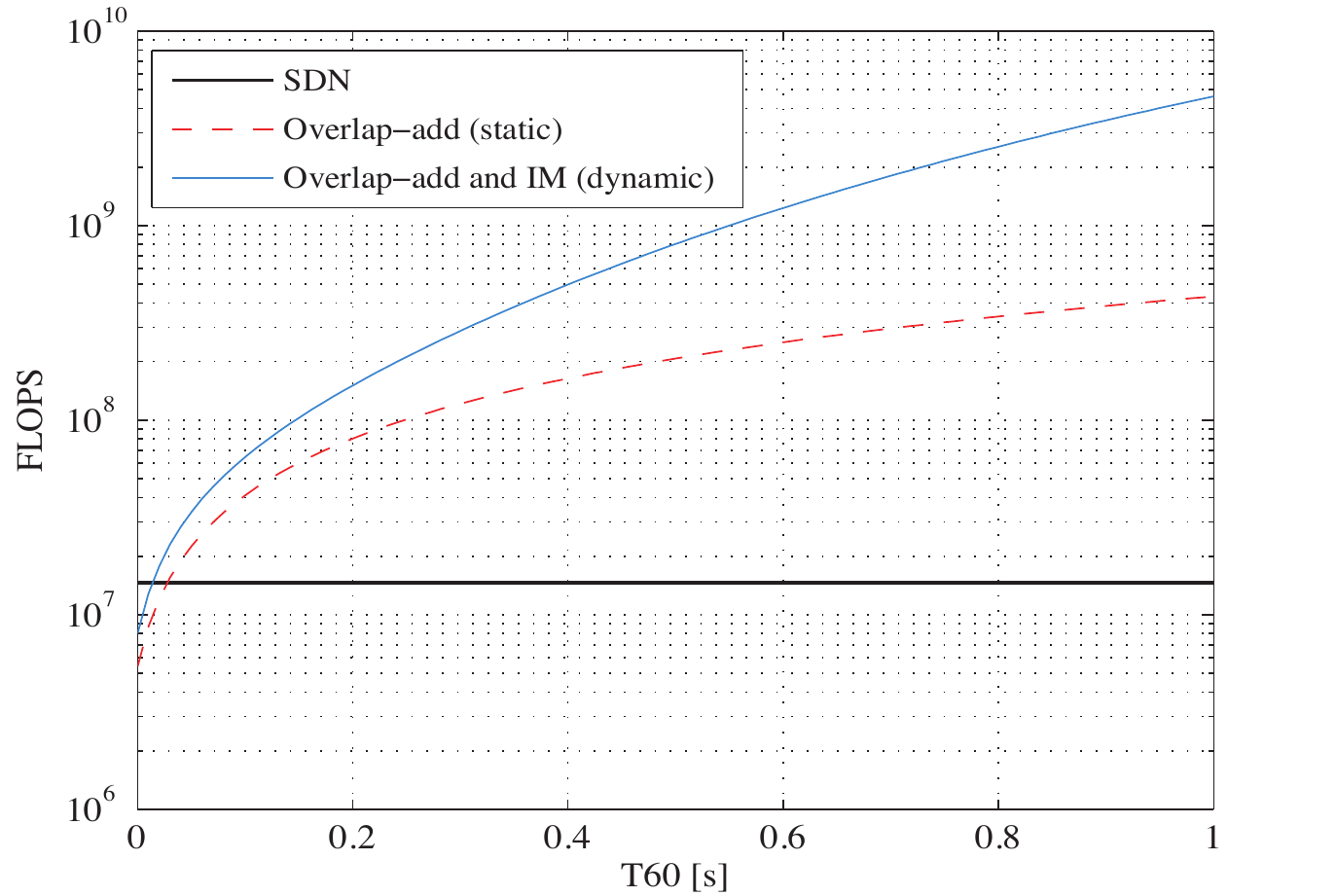}
\caption{Computational complexity of \protect\gls{sdn} in comparison to
overlap-add convolution in both static and dynamic modes as a function of reverberation time.
The static case represents the cost of the overlap-add convolution with a fixed, precomputed \protect\gls{rir}.
The dynamic case also includes the cost of calculating a new \protect\gls{rir}
and its \protect\gls{fft} for each time frame.
The sampling frequency is $F_s=44100$ Hz.}
\label{fig:plot_perf_1}
\end{figure}

\figurename~\ref{fig:plot_perf_1} compares the cost of \gls{sdn} with the
static and dynamic \gls{ism}.
In the static case, both microphone and source are not moving.
In the dynamic case, on the other hand, microphone and/or source are moving,
and the \gls{ism} is run for each frame.
The refresh rate is chosen such that the buffering delay is
shorter than the maximum latency of a half-frame delay between the video and audio, as recommended by the ITU Recommendation BR.265-9~\cite{itu265}.
For a video running at $25$ frames per second, 
this criterion gives a refresh rate of $F_r=50$ Hz.
The room size is the same as used in Allen and Berkley's paper of
$10 \times 15 \times 12.5$ feet~\cite{Allen79b}.
It may be observed that for typical medium-sized rooms, 
\gls{sdn} is from about $10$ to $100$ times faster than
dynamic \gls{ism}.
\gls{sdn} is also significantly faster than (overlap-add) convolution alone.

While we compare the computational complexity of the proposed algorithm with the standard overlap-add convolution, 
we also acknowledge that more efficient convolution methods have recently been proposed, e.g. \cite{Primavera:2014aa,atkins2013approximate,boholm2013}. 
However, the comparison of the proposed algorithm with these 
new methods is outside the scope of this article.

\subsection{Memory requirement}
The required memory  is determined by the number of taps of the
delay lines.
An upper bound for memory requirement $Q$ can be easily found by
observing that the length of each delay line is smaller than or equal
to the distance between the two farthest points of the simulated
space, giving:
\begin{equation}
Q\leq(N(N-1)+2N+1) \frac{qF_s}{c} R~~\text{bits}~,
\label{eq:mem}
\end{equation}
where $q$ is the number of bits per sample, and $R$ is the maximum
distance between any two points in the simulated space.  The value of
$R$ in the case of a rectangular room is $R=\sqrt{L_x^2+L_y^2+L_z^2}$.
For the more general case, $R$ is the diameter of the bounding sphere of the room shape.

Observe in (\ref{eq:mem}) that $Q$ scales linearly with the room size.
For a cubic room with a $5$~m edge, $F_s=40$~kHz, and $q=32$~bytes per
sample, the memory requirement is less than $170$~kB, which is
negligible for virtually every state-of-the-art platform.

\section{Assessment}\label{sec:results}
This section presents the results of assessments of \glspl{sdn} in terms of
perceptually-based objective criteria. 

As discussed in previous sections, 
first-order reflections are rendered correcty both in timing and amplitude by construction.
Another  cue important for the perception of room volume 
and materials is the reverberation time~\cite{rumsey2001spatial}.
Section~\ref{sec:rt} presents an evaluation of the reverberation time
both in frequency\nobreakdash-independent 
and frequency\nobreakdash-dependent cases.
Section~\ref{sec:echo} focuses on the
time evolution of echo density, 
which is related to the perceived 
time-domain texture of reverberation~\cite{huang2007aspects}.

Here, the \gls{ism} is used as a reference since it is 
the closest method among physical room models. 
More specifically, we use a C++ version of Allen and Berkley's
original time\nobreakdash-domain implementation~\cite{Allen79b}.

\subsection{Reverberation time}
\label{sec:rt}

The parameter most commonly  used to quantify the length of reverberation is $T_{60}$, 
which is defined as the time it takes for the room response 
to decay to 60 dB below its starting level.
In this section, the $T_{60}$ of the \gls{sdn} network 
is compared with 
two well-known empirical formulas 
proposed by Sabine~\cite{sabine1922collected} and Eyring \cite{eyring1930reverberation}:
\begin{align}
T_{60,Sab} &= \frac{0.161 V}{\sum_i A_i\alpha_i} \label{eq:sabine}~, \\
T_{60,Eyr} &= -\frac{0.161 V}{\left(\sum_i A_i\right)\log_{10}\left(1-\sum_i A_i\alpha_i/\sum_i A_i\right)}~, \label{eq:eyring}
\end{align}
where $V$ is the room volume, $A_i$ and $\alpha_i$ are 
the area and absorption coefficient of the $i$-th wall, respectively. 

\subsubsection{Frequency-independent wall absorption}

\begin{figure}
\centering
\includegraphics[width=0.5\textwidth]{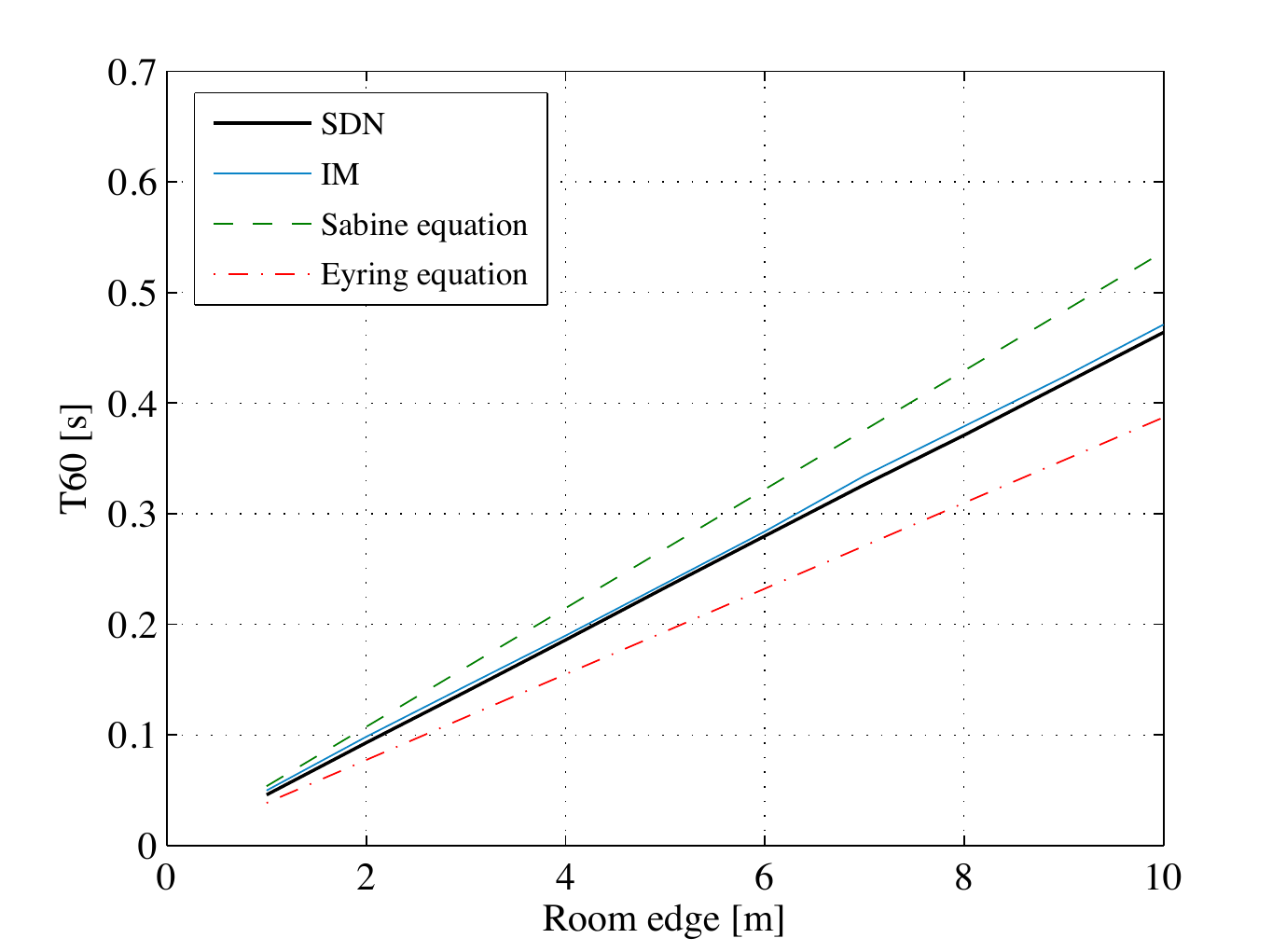}\label{fig:plot3_06}

\caption{Values of reverberation time $T_{60}$ as a function of the
  edge of a cubic room.  The wall absorption is $\alpha=0.5$, and both
  microphone and source are positioned in the volumetric center of the
  room.}\label{fig:plot3}
\end{figure}

Cubic rooms with different volumes and uniform frequency-independent
absorption are simulated.  In order to maintain the experimental 
conditions  across different room sizes, both the source and
the microphone are placed at the volumetric center of the room.  Furthermore, the
line-of-sight component was removed, as suggested in the
ISO 3382 standard \cite{standard19973382} for measuring the reverberation
time in small enclosures.  In \figurename~\ref{fig:plot3}, the
reverberation time is shown as a function of the edge length for a
room absorption coefficient $\alpha=0.5$.  It may be observed that the \gls{sdn}
generates room impulse-responses which  have reverberation times that
increase linearly with the edge length.  This is due to the larger
average distance between the nodes, 
which in turn increases the  mean free path of the structure.  

The $T_{60}$ values corresponding to the \gls{sdn} reverberator
are between the predictions given by Sabine and Eyring's
formulas and are nearly identical to the ones produced by the \gls{ism}.
The latter result may seem surprising if one considers that the  \gls{sdn}, as
opposed to the \gls{ism}, does 
not include attenuation due to spherical spreading 
(except for the initial first-order reflections and microphone taps).
This apparent inconsistency is resolved intuitively by observing that 
spherical spreading is a lossless phenomenon:
In the \gls{ism}, the quadratic energy decrease ($1/r^2$) is compensated 
by the quadratic increase of the number of contributing image sources over time, and
similarly to that,
in \gls{dwn} and \gls{sdn}, ``plane waves'' are scattered losslessly at each node, thus conserving the energy.

In \figurename~\ref{fig:rt60vsalpha}, the reverberation time is shown
as a function of the absorption coefficient $\alpha$.  The enclosure was taken
as a cubic room with a $5$~m edge, and results were averaged across 10 pairs of
source-microphone positions.
The coordinates were taken
from a uniform random distribution and satisfied both requirements set
in \cite{standard19973382}: The microphone was at least $1$ m away
from the nearest wall, and the distance between the source and microphone was
at least 
\begin{equation}
d_{min}=2\sqrt{\frac{V}{cT_{est}}}~,
\label{eq:ISOdist}
\end{equation}
where $T_{est}$ is a coarse estimation of the reverberation time.  In
these simulations, $T_{est}$ was set using Sabine's formula.

\begin{figure}
\centering
\includegraphics[width=0.5\textwidth]{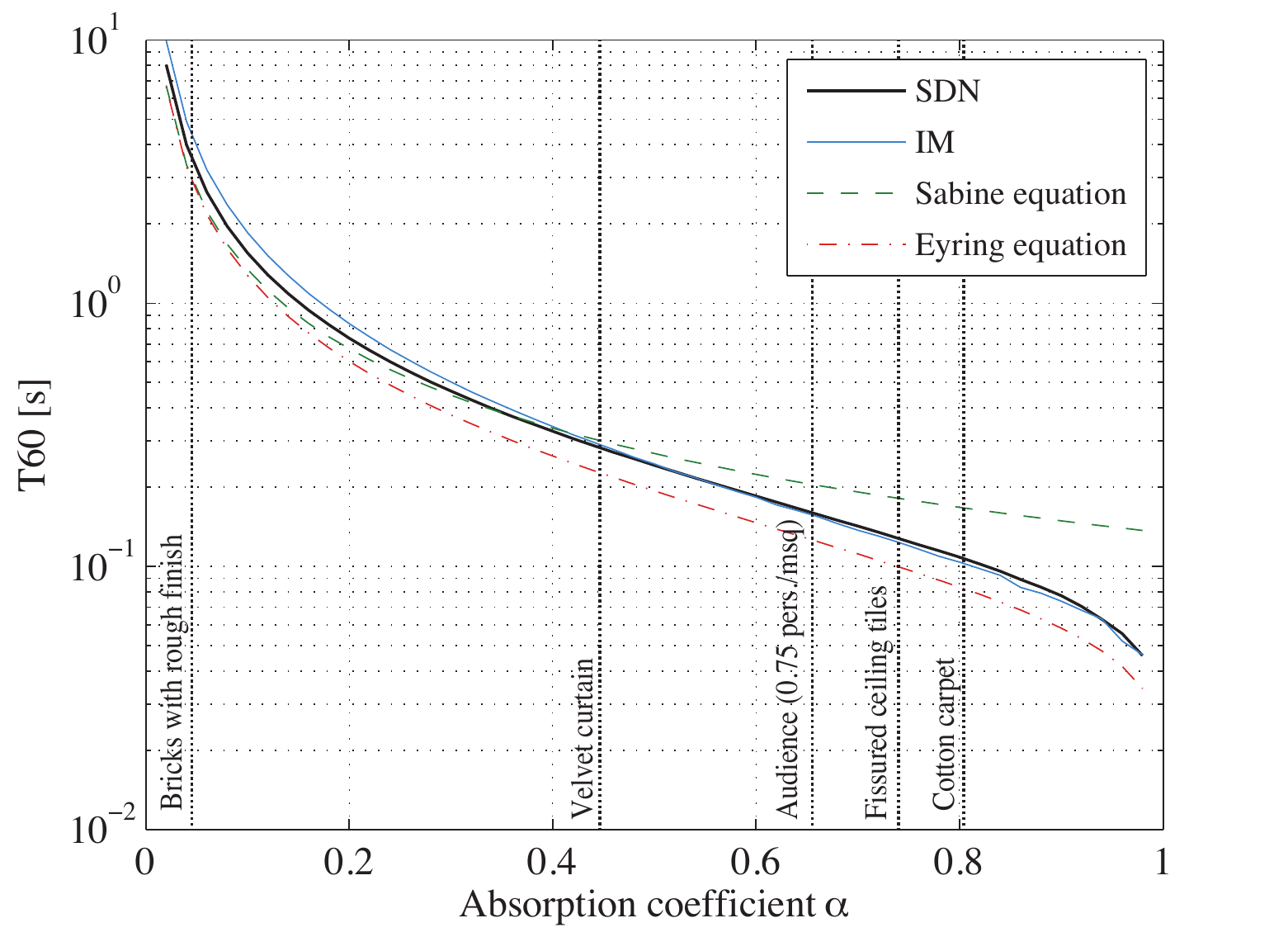}
\caption{Reverberation time values as a function of the absorption
  coefficient $\alpha$ for \protect\gls{sdn}, \protect\gls{ism}, and Sabine and Eyring
  predictions. The simulated enclosure is a cube with $5$ m edge.
  $T_{60}$ values are averaged across 10 source and microphone random
  positions.  The $1$~kHz absorption coefficient of various materials
  as measured by Vorl{\"a}nder in \cite{vorlander2008auralization} are
  reported at the bottom of the plot as reference.}
\label{fig:rt60vsalpha}
\end{figure}

It may be observed that for absorption coefficients higher than around
$\alpha=0.4$, \gls{sdn} and the \gls{ism} are nearly identical and are
both between Sabine and Eyring's formulas.  For high absorption
coefficients, both \gls{sdn} and the \gls{ism} approach Eyring's formula,
which is known to give more accurate predictions in that
region~\cite{eyring1930reverberation}.  For low absorption
coefficients, \gls{sdn} and the \gls{ism} produce reverberation times
longer than Sabine and Eyring's formulas, with \gls{sdn} being closer
to both.

\subsubsection{Frequency-dependent wall absorption}

\begin{figure}
\centering
\includegraphics[width=0.45\textwidth]{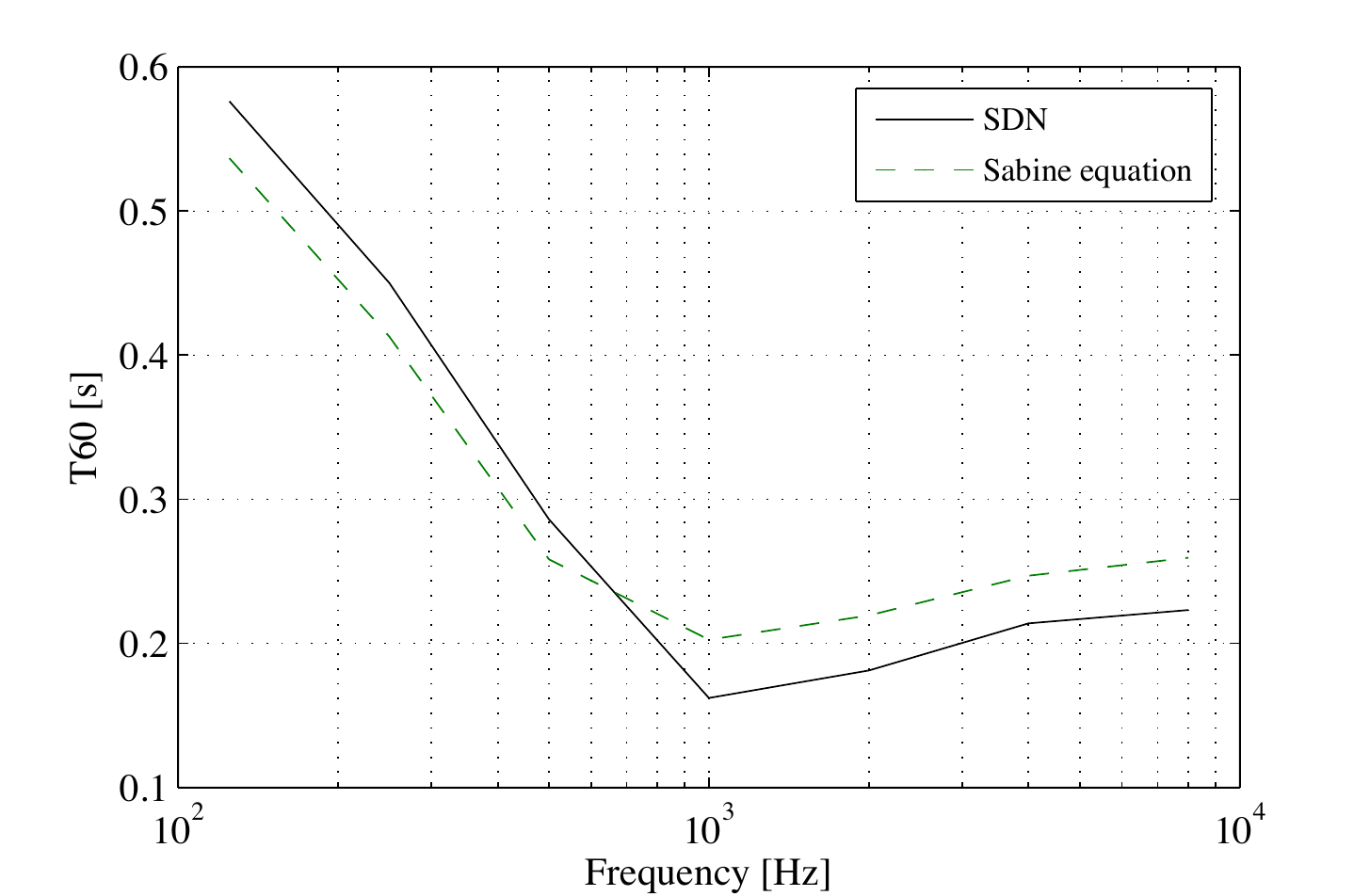}
\caption{Comparison of reverberation time in different octave bands
  for SDN and Sabine's formula prediction.}\label{fig:filters}
\end{figure}

\figurename~\ref{fig:filters} shows the result of a simulation where
all walls mimic the frequency-dependent absorption of cotton carpet.
The filters $H_i(z)$ are all set to be equal to a filter $H(z)$ which was implemented as a minimum-phase IIR filter with
coefficients optimized by a damped Gauss-Newton method to fit the
absorption coefficients reported by Vorl{\"a}nder in
\cite{vorlander2008auralization}.
This procedure gave
$$H(z)=\frac{0.6876 -1.9207 z^{-1}+ 1.7899z^{-2}
  -0.5567z^{-3}}{1 -2.7618z^{-1}+ 2.5368z^{-2}
  -0.7749z^{-3}}~.$$ 
  The  source and microphone were positioned on the diagonal
of a cubic room with  $5$ m edge.
More specifically, they were positioned on the diagonal at a distance of
$d_{min}=2.96$ m from the center, as specified using (\ref{eq:ISOdist}).

In \figurename~\ref{fig:filters} the wall filter response is plotted
together with the corresponding Sabine predictions in (\ref{eq:sabine}), which for the given
room becomes
\begin{equation}
T_{60,Sab}= \frac{0.161 V}{\sum_i
  A_i\alpha_i(\omega)}=\frac{0.161\cdot 5}{6
  \alpha(\omega)}=\frac{0.161\cdot 5}{6 (1-|H(e^{j\omega})|^2)}~.
\end{equation}
The simulated \gls{rir} was fed into an octave-band filter bank, and
$T_{60}$ values were calculated for each octave-band.  As shown in
\figurename~\ref{fig:filters}, these measured $T_{60}$ are very close
to Sabine's formula prediction, thus confirming that the proposed
model allows controlling the absorption behavior of wall materials
explicitly.

Note that if explicit control of the reverberation time is sought, 
the prediction functions (\ref{eq:sabine}) or (\ref{eq:eyring}) 
can be inverted to obtain  needed $\alpha(\omega)$.
To this end, Eyring's formula (\ref{eq:eyring}) is preferable, 
since (\ref{eq:sabine}) may yield non-physical values for the absorption coefficient (i.e.~$\alpha>1$) 
of some acoustically ``dead" rooms \cite{eyring1930reverberation}.

\subsection{Echo density}
\label{sec:echo}

The time evolution of echo density 
is commonly thought to influence the 
perceived time-domain \emph{texture} of reverberation~\cite{huang2007aspects}.
In an effort to quantify this perceptual attribute, 
Abel and Huang defined the \gls{ned},
which was found to have a very strong correlation with 
results of listening tests~\cite{huang2007aspects}.
The \gls{ned} is defined 
as the percentage of samples of the room impulse response lying 
more than a standard deviation away from the mean 
in a given time window compared to that expected for Gaussian noise.
A \gls{ned} equal to $1$ means that, within the considered window,  
the number of 
samples lying 
more than one standard deviation away from the mean is equal to the  
one observed with Gaussian noise.

\figurename~\ref{fig:echo} shows the time evolution of the \gls{ned} 
obtained with the \gls{ism} and with the proposed model using 
three different scattering matrices.
The scattering matrices are (a) the isotropic matrix, (b) a random orthogonal matrix, and (c) a random permutation matrix. 
The random orthogonal matrix was obtained by setting 
the angles of a Givens-rotation 
parametrization of orthogonal matrices~\cite{vetterli1995wavelets} at random.
The simulated enclosure was a rectangular room with 
dimensions $3.2\times 4.0 \times 2.7$ m, 
and results were averaged across $50$ random 
pairs of source and microphone positions.
The wall gains were set to~$\beta=-\sqrt{0.9}$ (wall absorption of $\alpha=0.1$),
with the negative sign being chosen in order to obtain a zero-mean 
reverberation tail with the \gls{ism}.

\figurename~\ref{fig:echo} shows that the 
build-up of echo density of \gls{sdn} is
very close to that of the \gls{ism} when the isotropic scattering matrix is employed.
In particular, 
the \gls{ned} values of $0.3$ and $0.75$, 
which were previously identified 
as breakpoints dividing three perceptually distinct groups~\cite{huang2007aspects}, 
are reached at around the same delays by the two methods. 
Notice how the permutation matrix 
fails to reach a Gaussian-like reverberation.
The random orthogonal matrix, on the other hand, 
does reach a Gaussian-like reverberation, 
but it takes longer to achieve the desired reverberation quality characterized by the $0.75$ breakpoint.

\begin{figure}
\centering
\includegraphics[width=0.5\textwidth]{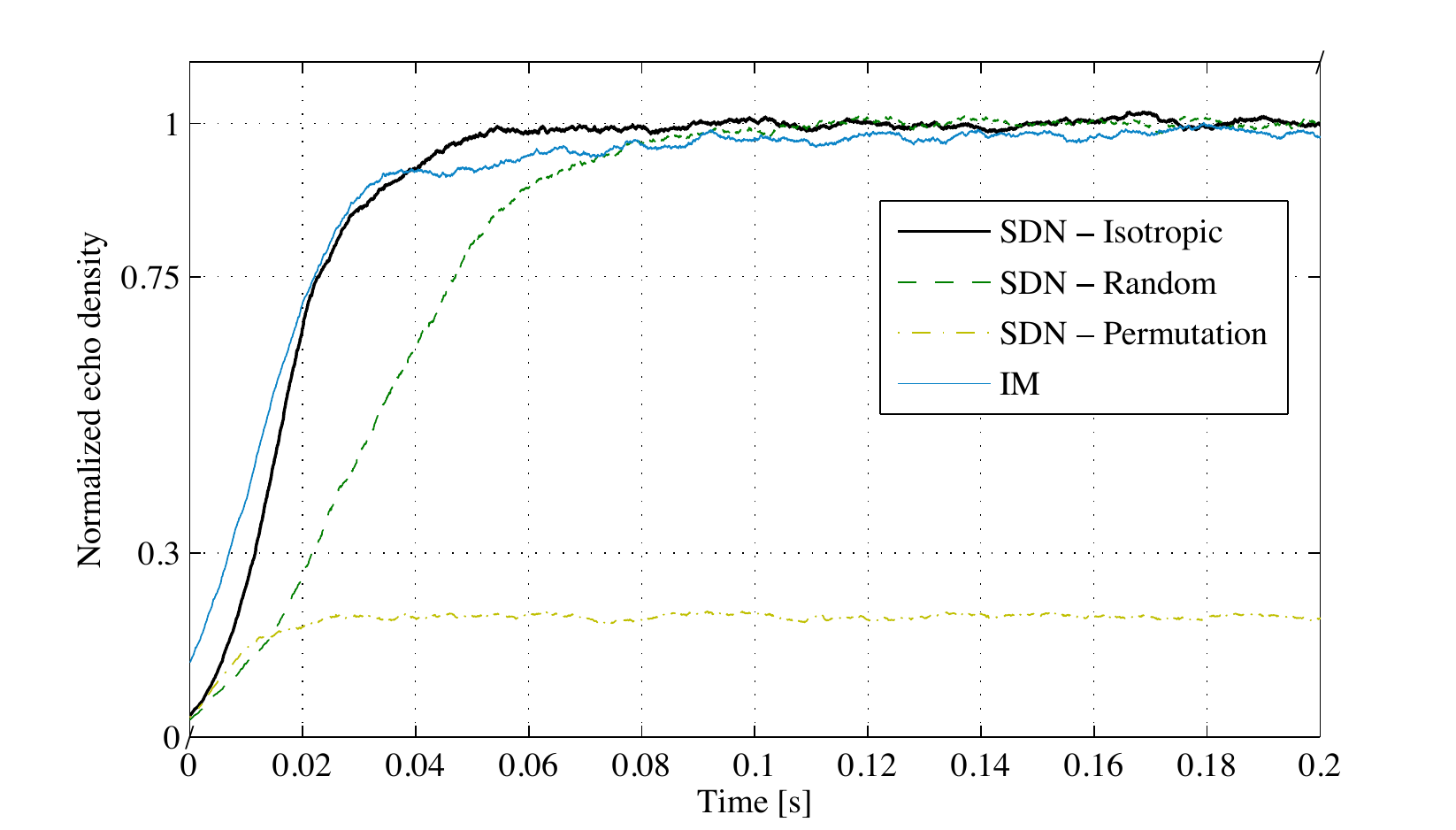}
\caption{Time evolution of the normalized echo density for the \protect\gls{ism}
  and \protect\gls{sdn} with various scattering matrices $\mathbf{S}$. }
\label{fig:echo}
\end{figure}

\subsection{Mode density}

The mode density, i.e. the average number of resonant frequencies per Hz, 
is another important perceptual property in artificial reverberation~\cite{jot1991}.
In order to achieve a natural-sounding reverberation, 
the mode density should be sufficiently high, 
such that no single resonance stands out causing metallic-sounding artifacts. 
A threshold for the minimum mode density that is commonly used in this context is~\cite{jot1991}:
\begin{equation}
d_{min}=\frac{T_{60}}{6.7}~,
\end{equation}
which is due to an early work of Schroeder~\cite{schroeder1962frequency}.

The mode density in \gls{sdn} can be calculated  
using considerations similar to those used for \glspl{fdn}~\cite{smith2010physical,jot1991}. 
Assuming that the wall filters are simple gains, i.e. $\overline{\mathbf{H}}(z)=\beta\mathbf{I}$, 
and applying the inversion lemma to the transfer function~(\ref{eq:transferfunct}), 
the poles of the system can be calculated as the solutions of
\begin{equation}
\det\left(\mathbf{D}_f(z^{-1})-\beta\overline{\mathbf{A}}\mathbf{P}\right)=0~.
\label{eq:determinant}
\end{equation}
Using Leibniz's formula for determinants it is easy to see that 
the order of the polynomial in (\ref{eq:determinant}) is equal to the summation of all the delay-line lengths 
in the \gls{sdn} backbone, i.e. $\sum_{i}\sum_{j \neq i}D_{i,j}$.
Using the fundamental theorem of algebra, and assuming that the poles are uniformly distributed~\cite{smith2010physical}, 
the mode density of the \gls{sdn} network can be calculated as
\begin{equation}
d_f=\frac{1}{F_s}\sum_{i}\sum_{j \neq i}D_{i,j}~.
\label{eq:schroeder}
\end{equation}

The \gls{sdn} structure satisfies $d_f>d_{min}$ under most conditions of practical interest.
This can be easily shown analytically in the case where
source and microphone are close to the volumetric center of a cubic room with edge $L$. 
In this case, the length of the delay lines is approximately $L \frac{F_s}{c}$ for the six lines
connecting opposite walls and $\frac{\sqrt{2}}{2} L \frac{F_s}{c}$ for the remaining twenty four lines.
The mode density is thus 
\begin{equation}
d_f = \left(6+24\frac{\sqrt{2}}{2}\right) \frac{L}{c}\approx 23 \frac{L}{c}~.
\label{eq:modedensity}
\end{equation}
The condition $d_f>d_{min}$ is then satisfied whenever 
\begin{equation}
L > \frac{c}{6.7 \times 23} T_{60} \approx 2.22~T_{60}~.
\label{eq:edgemin}
\end{equation}
Since practical $T_{60}$ values for reverberation are on the order of a second, 
it follows that \glspl{sdn} 
has a sufficient mode density as long as $L>2.22$ m 
(i.e. volume larger than around $11~\text{m}^3$), 
which covers most cases of practical interest.

By replacing $T_{60}$ in (\ref{eq:edgemin}) with Sabine's approximation (\ref{eq:sabine}),
it can also be seen that 
the condition $d_f>d_{min}$ is satisfied whenever 
the absorption coefficient 
is larger than $\alpha>0.06$, or, equivalently, $\beta<0.97$.

In order to assess the mode density in cases more general than the cubic one,
Monte Carlo simulations were run using rectangular rooms 
with randomly selected parameters.
More specifically, the three room dimensions 
were each drawn from a uniform distribution between $2$ and $10$ meters.
The absorption coefficient was common to all walls 
and was drawn from a uniform distribution between $0$ and $1$.
The microphone and the source were placed in positions chosen at random 
within the room boundaries. 
Out of 1000 simulations, \gls{sdn} provided a sufficient mode density in the $94.7$\% of cases.
Among the cases that did not provide sufficient echo density, 
the largest absorption coefficient was $\alpha=0.066$,
which is largely in agreement with the result obtained above for cubic rooms.

\section{Conclusions and Future Work}\label{sec:conclusions}
\glsreset{sdn}

This paper presented a scalable and interactive artificial reverberator 
termed \gls{sdn}.
The room is modeled by
scattering nodes interconnected by bidirectional delay lines. These
scattering nodes are positioned at the points where first-order
reflections originate. In this way, the first-order reflections are
simulated correctly, while a rich but less accurate reverberation tail
is obtained. 
It was shown that, according to various objective measures of perceptual features, 
\gls{sdn} achieves a reverberation quality similar to that of the \gls{ism} 
while having a computational load one to two orders of magnitude lower.

The interested reader can listen to \gls{sdn}-generated audio samples
that are made available as supplementary downloadable material 
at http://ieeexplore.ieee.org and at~\cite{desenaorgsdn}.

Several directions for future research can be envisioned on the basis  
of the work presented in this paper. 
The design of the backbone network in \figurename~\ref{fig:SDN} 
has a simple geometrical interpretation.
However, the lengths of the delay lines in the backbone network can 
be designed using different approaches.
It is believed, in fact, that, as long as the network 
has a mean-free-path similar to the one of the corresponding physical space,
the \gls{sdn} will conserve its appealing properties observed in Sec.~\ref{sec:results}.
The delay-line lengths could be designed, for instance, 
to minimize the average timing error of higher-order reflections 
or to achieve a better fit with the modal frequencies of the physical space. 
The number of \gls{sdn} nodes could also be increased.
Indeed, 
while using a single \gls{sdn} node for each wall 
is the minimum to ensure that all higher-order reflections are modeled, 
a larger number of nodes could be used, for instance, to further
increase the modal density or to emulate higher-order reflections exactly.
This would, of course, come at the expense of an increased computational complexity.

\section*{Appendix}
\label{app:matrix}

\noindent Towards proving Theorem~\ref{theorem:param}, we first establish the following lemma.
\begin{lemma}
Two  diagonalizable matrices have the same eigenvalues if and only if they are similar.
\label{lem:eigensimilar}
\end{lemma}
\begin{proof}[Proof of Lemma \ref{lem:eigensimilar}]
The sufficiency is a well know result on similar matrices~\cite{hogben2006handbook}.
To prove the necessity, let us consider
 two diagonalizable matrices $\mathbf{A}$ and~$\mathbf{B}$ which have the same eigenvalues. 
Since $\mathbf{A}$ is diagonalizable, it follows that $\mathbf{V}_A^{-1}\mathbf{A}\mathbf{V}_A=\mathbf{\Lambda}$, where
$\mathbf{\Lambda}$ is a diagonal matrix and $\mathbf{V}_A$ is an invertible matrix. Hence, the following equality holds:
 $\mathbf{A}\mathbf{V}_A=\mathbf{V}_A\mathbf{\Lambda}$, which implies that $\mathbf{\Lambda}$ has eigenvalues
 of $\mathbf{A}$ on its diagonal. Since $\mathbf{B}$ is also diagonalizable and has the same eigenvalues as $\mathbf{A}$, it
 satisfies $\mathbf{V}_B^{-1}\mathbf{B}\mathbf{V}_B=\mathbf{\Lambda}$.  Thus, 
$\mathbf{V}_A^{-1}\mathbf{A}\mathbf{V}_A=\mathbf{\Lambda}=\mathbf{V}_B^{-1}\mathbf{B}\mathbf{V}_B$, which further implies that 
$\mathbf{A}=\left(\mathbf{V}_B\mathbf{V}_A^{-1}\right)^{-1}\mathbf{B}\left(\mathbf{V}_B\mathbf{V}_A^{-1}\right)$, {\sl i.e.} the two matrices are similar.
\end{proof}
Using Lemma \ref{lem:eigensimilar} we can now prove 
Theorem~\ref{theorem:param}:
\begin{proof}[Proof of Theorem~\ref{theorem:param}]
First, observe that $\mathbf{\Lambda}$ is block diagonal, and that eigenvalues of each block are mutually distinct. 
Hence, each block of  $\mathbf{\Lambda}$ is diagonalizable, and therefore
$\mathbf{\Lambda}$ is itself diagonalizable~\cite{hogben2006handbook} 
 (over the field of complex numbers, $\mathbb{C}$). Thus, $\mathbf{\Lambda}$ 
 has the same eigenvalues as $\mathbf{A}$ and is diagonalizable.
Since $\mathbf{\Lambda}$ and $\mathbf{A}$ are both diagonalizable, and have the same eigenvalues, according to 
Lemma \ref{lem:eigensimilar} they must be similar over $\mathbb{C}$.
Moreover, since both  $\mathbf{\Lambda}$ and $\mathbf{A}$ are real and similar over $\mathbb{C}$, 
they must be also similar over $\mathbb{R}$~\cite{hogben2006handbook}, that is,
there must exist a real invertible matrix $\mathbf{T}$ such that $\mathbf{A}=\mathbf{T}^{-1}\mathbf{\Lambda}\mathbf{T}$.
On the other hand if $\mathbf{A}$ has the form given in (\ref{eq:aform}), it is diagonalizable since $\mathbf{\Lambda}$ is diagonalizable.
\end{proof}

\begin{proof}[Proof of Theorem \ref{theorem:houseeigen}]
Substituting the definition of a Householder transformation into the
cost function in (\ref{eq:scatproblem}), here termed $\Phi(\vec{v})$,
leads with simple algebraic manipulations to
$\Phi(\vec{v})=\left(\vec{v}^T\vec{v}-1\right)\vec{v}^T\vec{v}-\vec{v}^T\mathbf{D}\vec{v}+\text{const}$.
Minimizing $\Phi(\vec{v})$ subject to
$\vec{v}^T\vec{v}=1$  is therefore equivalent to maximising 
$\Phi_1(\vec{v})=\vec{v}^T\mathbf{D}\vec{v}$
under the same constraint. The new cost function can be further expressed as
$$\Phi_1(\vec{v})=\frac{1}{2}\left(\vec{v}^T\mathbf{D}\vec{v}+\vec{v}^T\mathbf{D}^T\vec{v}\right)=
\frac{1}{2}\vec{v}^T\left(\mathbf{D}+\mathbf{D}^T\right)\vec{v}
$$
and the unit norm vector which maximise it, is therefore the singular vector 
which corresponds to the largest singular value of $\mathbf{D}+\mathbf{D}^T$.
\end{proof}

\begin{proof}[Proof of Theorem \ref{theorem:isotropic}]
The property that $\mathbf{S}$ scatters energy from each node equally
among all other nodes requires that all off-diagonal elements are identical, and thus
that $\mathbf{A}$ has the form

\begin{equation}
\arraycolsep=5.4pt\def\arraystretch{1.0}
\mathbf{A}=
  \begin{bmatrix}
    a_1 & a_0  & \cdots & a_0\\
    a_0 & a_2  & \cdots & a_0\\
    \cdots & \cdots  & \cdots & \cdots \\
    a_0 & a_0 & \cdots & a_K\\
  \end{bmatrix}~.
\end{equation}
Orthogonality of $\mathbf{A}$ requires that $a_0, \dots,a_K$ satisfy
\begin{equation}
\left\{
\begin{aligned}
&a_i^2+(K-1)a_0^2=1~~~~~~~~~~~~~~i=1,\dots,K \\
&a_0(a_i+a_j)+(K-2)a_0^2=0~~~\forall i\neq j 
\end{aligned}
\right.~.
\label{eq:scatteringsystem}
\end{equation}
The first constraint can be written as $a_i=\pm\sqrt{1-(K-1)a_0^2}$,
which implies that all diagonal elements are identical in magnitude.
The second constraint implies that if $a_0\neq 0$
(the solution $a_0=0$ is ignored since it does not scatter energy), 
the diagonal elements have also the same sign.
Solving (\ref{eq:scatteringsystem}) with $a_1=\dots=a_K$ 
yields $a_1= \pm(2-K)/K$ and $a_0=\pm 2/K$, thus proving the theorem.
\end{proof}

\section*{Acknowledgments}
The authors would like to thank Jonathan Abel and Geoffrey Robinson 
for their help and support.

\bibliographystyle{IEEEtran}
\bibliography{library}

\begin{IEEEbiography}[{\includegraphics[width=1in,height=1.25in,clip,keepaspectratio]{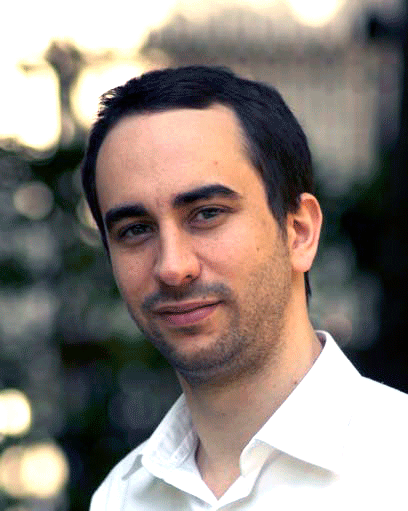}}]{Enzo De Sena}

(S'11-M'14) received the B.Sc. in 2007 and M.Sc. \emph{cum laude} in 2009, both from the Universit\`a degli Studi di Napoli "Federico II" in Telecommunication Engineering. 
In 2013, he received the Ph.D. degree from King's College London in Electronic Engineering.
He is currently a Postdoctoral Research Fellow at KU Leuven.  
From 2012 to 2013 he was a Teaching Fellow at King's College London.
From 2013 to 2015 he was a Marie Curie Fellow in the  
``Dereverberation and Reverberation of Audio, Music, and Speech" ITN at KU Leuven.
He previously collaborated with the Network Research Lab at UCLA (2007-2009),  
and he was a Visiting Researcher at 
the Center for Computer Research in Music and Acoustics at Stanford University (2013)
and at the Signal and Information Processing section at Aalborg University (2014-2015).
His current research interests include room acoustics modelling, 
multichannel audio systems, microphone beamforming and binaural modelling.
He is a member of IEEE, EURASIP,  
and the Acoustical Society of America. 
\end{IEEEbiography}

\begin{IEEEbiography}[{\includegraphics[width=1in,height=1.25in,clip,keepaspectratio]{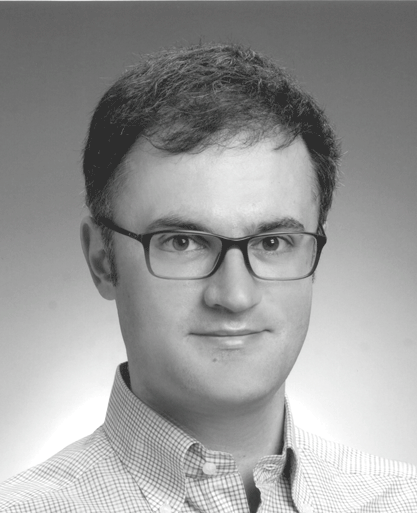}}]{H\"useyin Hac\i{}habibo\u{g}lu}
(S'96-M'00-SM'12) received the B.Sc. (honors) degree from the Middle East Technical University (METU), Ankara, Turkey, in 2000, the M.Sc. degree from the University of Bristol, Bristol, U.K., in 2001, both in electrical and electronic engineering, and the Ph.D. degree in computer science from Queen's University Belfast, Belfast, U.K., in 2004. He held research positions at University of Surrey, Guildford, U.K.\ (2004-2008) and King's College London, London, U.K.\ (2008-2011). Currently, he is an Associate Professor and Head of Department of Modelling and Simulation at Graduate School of Informatics, Middle East Technical University, Ankara, Turkey. His research interests include audio signal processing, room acoustics, multichannel audio systems, psychoacoustics of spatial hearing, microphone arrays, and game audio. Dr.\ Hac\i{}habibo\u{g}lu is a member of the IEEE Signal Processing Society, Audio Engineering Society (AES),  Turkish Acoustics Society (TAD), and the European Acoustics Association (EAA).
\end{IEEEbiography}

\begin{IEEEbiography}[{\includegraphics[width=1in,height=1.25in,clip,keepaspectratio]{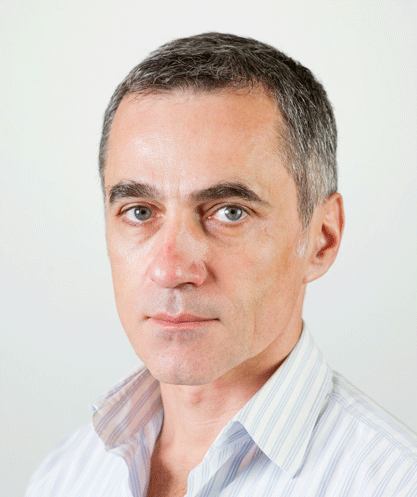}}]{Zoran Cvetkovi\'c}
(SM'04) is Professor of  Signal Processing at King's College London. He received his Dipl. Ing. and Mag. degrees from the University of Belgrade, Yugoslavia,  the M.Phil. from Columbia University, and the Ph.D. in electrical engineering from the University of California, Berkeley. He held research positions at EPFL, Lausanne, Switzerland (1996), and at Harvard University (2002-04). Between 1997 and 2002 he was a member of the technical staff of AT\&T Shannon Laboratory. His research interests are in the broad area of signal processing, ranging from theoretical aspects of signal analysis to applications in audio and speech technology, and neuroscience. From 2005 to 2008 he served as an Associate Editor of IEEE Transactions on Signal Processing.
\end{IEEEbiography}

\begin{IEEEbiography}[{\includegraphics[width=1in,height=1.25in,clip,keepaspectratio]{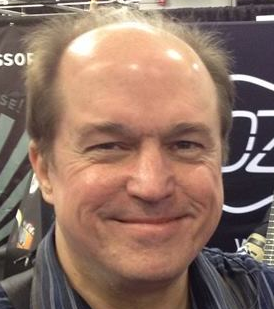}}]{Julius O.~Smith} received the B.S.E.E. degree from Rice
University, Houston, TX, in 1975 (Control, Circuits, and
Communication).  He received the M.S. and Ph.D. degrees in E.E. from
Stanford University, Stanford, CA, in 1978 and 1983, respectively.
His Ph.D. research was devoted to improved methods for digital filter
design and system identification applied to music and audio systems.
From 1975 to 1977 he worked in the Signal Processing Department at
ESL, Sunnyvale, CA, on systems for digital communications.  From 1982
to 1986 he was with the Adaptive Systems Department at Systems Control
Technology, Palo Alto, CA, where he worked in the areas of adaptive
filtering and spectral estimation.  From 1986 to 1991 he was employed
at NeXT Computer, Inc., responsible for sound, music, and signal
processing software for the NeXT computer workstation.  After NeXT, he
became an Associate Professor at the Center for Computer Research in
Music and Acoustics (CCRMA) at Stanford, teaching courses and pursuing
research related to signal processing techniques applied to music and
audio systems.  Continuing this work, he is presently Professor of
Music and (by courtesy) Electrical Engineering
at Stanford University.  For more information, see
http://ccrma.stanford.edu/\~{}jos/.
\end{IEEEbiography}

\end{document}